%% file: paper.tex
\documentclass[a4paper,USenglish]{lipics}

\usepackage{comment}
\usepackage{amsmath}
\usepackage{amssymb}
\usepackage{graphicx}
\usepackage{epstopdf}
\usepackage{caption}
\usepackage{url}
\usepackage{wrapfig}
\newtheorem{invariant}{Invariant}

\newcommand{\mca}{\mathcal{A}}

\newcommand{\mcc}{\mathcal{C}}

\newcommand{\mcf}{\mathcal{F}}

\newcommand{\mci}{\mathcal{I}}

\newcommand{\mck}{\mathcal{K}}

\newcommand{\mcn}{\mathcal{N}}

\newcommand{\mcp}{\mathcal{P}}

\newcommand{\mcr}{\mathcal{R}}
\newcommand{\mcs}{\mathcal{S}}

\newcommand{\mcv}{\mathcal{V}}
\newcommand{\mcw}{\mathcal{W}}

\usepackage{microtype}%if unwanted, comment out or use option "draft"

\title{Moving Participants Turtle Consensus\footnote{This work was partially supported by
AFOSR DURIP grant FA2386-12-1-3008,
by NSF grants CCF-1047540, CNS-1040689, CNS-1422544, CNS-1561209, CNS-1601879,
by a Google Faculty Research Award,
by MDCN/iAd grant 54083,
and by gifts from Infosys, Facebook, and Amazon.com.  The authors would
also like to thank the anonymous reviewers of OPODIS.}}
\titlerunning{Moving Participants Turtle Consensus} %optional, in case that the title is too long; the running title should fit into the top page column

\author{Stavros Nikolaou}
\author{Robbert van Renesse}
\affil{Department of Computer Science, Cornell University, Ithaca, US\\
  \texttt{\{snikolaou,rvr\}@cs.cornell.edu}}
\authorrunning{S. Nikolaou and R. Van Renesse} %mandatory. First: Use abbreviated first/middle names. Second (only in severe cases): Use first author plus 'et. al.'

\Copyright{Stavros Nikolaou and Robbert van Renesse}%mandatory, please use full first names. LIPIcs license is "CC-BY";  http://creativecommons.org/licenses/by/3.0/

% mandatory: Please choose ACM 1998 classifications from http://www.acm.org/about/class/ccs98-html . E.g., cite as "F.1.1 Models of Computation". 
% mandatory: Please provide 1-5 keywords
\begin{document}

% Hyphenation

\maketitle

% Abstract
\input{abstract}

% Introduction
\input{introduction}

% System Model
\input{model}

% Algorithms
\input{algorithms}
\input{implementation}

% Evaluation
\input{evaluation}
\input{relatedwork}

% Conclusions
\input{conclusions}
\bibliographystyle{plainurl}
\bibliography{mybib}  % sigproc.bib is the name of the Bibliography in this case

\newpage

\appendix

% MPTC Correctness
\input{correctness}

%% BMPTC Correctness
%\input{byzantine_correctness}

% Byzantine extension
\input{byzantine}

% Implementation details
\input{implementation_transitions}

\end{document}

%% file: abstract.tex
\begin{abstract}

We present Moving Participants Turtle Consensus (MPTC), an asynchronous
consensus protocol for crash and Byzantine-tolerant distributed
systems.  MPTC uses various \emph{moving target defense} strategies
to tolerate certain Denial-of-Service (DoS) attacks issued by an
adversary capable of compromising a bounded portion of the system.
MPTC supports on the fly reconfiguration of the consensus strategy
as well as of the processes executing this strategy when solving
the problem of agreement.  It uses existing cryptographic techniques
to ensure that reconfiguration takes place in an unpredictable
fashion thus eliminating the adversary's advantage on predicting
protocol and execution-specific information that can be used against
the protocol.

We implement MPTC as well as a State Machine Replication protocol
% based on this MPTC implementation
and evaluate our design
under different attack scenarios. Our evaluation shows that
MPTC approximates best case scenario performance
even under a well-coordinated DoS attack.

%In this chapter we describe an extension to the previously described Turtle Consensus protocol. This extension, that we call Moving Participants Turtle Consensus, enables the protocol to not only change the consensus strategy on the fly but also the set of processes that execute that strategy and it does so in an unpredictable fashion using cryptographic techniques.

\end{abstract}

%% file: introduction.tex
\section{Introduction}\label{sec:mptc_intro}

Most distributed systems today are designed to tolerate failures. Existing 
fault-tolerance methods typically assume that failures are rare.
They are tailored to provide good performance when no failures occur but might perform poorly 
under failure scenarios. However, as shown in works like~\cite{CWAD09}, such 
designs allow malicious adversaries to craft workloads and Denial-of-Service (DoS) attacks that can substantially degrade the performance of certain 
state-of-the-art fault-tolerance protocols. As such DoS attacks become more 
common, it is becoming increasingly important to design fault-tolerance 
mechanisms that perform well in good scenarios while also 
gracefully handle adversarial ones. A core building block of many of these 
mechanisms are consensus protocols used by a set of replicas 
to agree on some state. One way to improve existing fault tolerance solutions 
is by enhancing the underlying consensus protocols with reconfiguration 
capabilities that allow them to change their execution parameters on the fly in 
order to better deal with adversarial workloads.

Our prior work on \emph{Turtle Consensus}~\cite{NvR15}
also aims at attack-tolerant consensus.
Turtle Consensus is a round-based consensus protocol that operates by using 
different consensus strategies across different rounds.
The system's processes try to reach agreement running a round of some consensus 
protocol in the literature; if they fail to do so they move onto the next 
round using a different protocol.
The selection of each round's strategy is predetermined 
and known to all processes running the protocol.
%Turtle Consensus employs fine-grained 
%reconfiguration that significantly changes the attack surface of the protocol, 
%thus making it more robust against protocol-specific attacks. 
The main strategy of Turtle Consensus is to use the best approach available for normal operation in a particular setting and switch to different ``backup protocols'' as soon as the approach becomes inefficient, for example in the case of a DoS attack.
We showed that the approach used sub-optimal strategies 
during DoS attacks, thus bounding the protocol's efficiency to the capabilities 
of these backup protocols.
%In the evaluation of Turtle Consensus it became apparent that this forces the protocol to sub-optimal strategies when under attack, thus bounding the protocol's efficiency to the capabilities of these back-up protocols. 
In addition, an adversary capable of compromising even a single consensus 
participant can learn and use the predetermined nature of the protocols' 
succession to constantly drive the system to sub-optimal executions.

In this paper we address these concerns by adding another degree of 
freedom in the reconfiguration capabilities of Turtle Consensus. We present 
Moving Participants Turtle Consensus (MPTC), an extension to the Turtle 
Consensus protocol that allows switching not only the protocols but also the 
set of processes on which they run across different rounds of a single 
consensus instance. The consensus protocol round and the processes 
participating in its execution form what we call a \emph{configuration}, which 
our approach changes unpredictably at each round. While the configuration 
selection for each round is predetermined by a trusted dealer, it is unknown to 
the processes during MPTC execution. Using existing cryptographic techniques, 
we ensure that, only if sufficiently many processes collaborate during some 
round, the next round's configuration can be determined. This renders MPTC a 
valuable tool for building systems that can tolerate DoS attacks in both crash-
and Byzantine-tolerant environments where a bounded portion of the system may 
be compromised.

\begin{comment}
The paper is organized as follows: In Section~\ref{sec:mptc_model} we 
describe the system model and provide some background on the cryptographic 
primitives we use as well as consensus protocols. In 
Section~\ref{sec:mptc_protocol} we present our proposed MPTC protocol for crash 
failures and demonstrate its correctness. (We extend this description to 
address byzantine failures in Appendix~\ref{sec:mptc_byzantine}.) Then in 
Section~\ref{sec:mptc_implementation} we present an implementation of MPTC and 
describe how we used it to implemented a state machine replication protocol. We 
describe our experimental results in Section~\ref{sec:mptc_evaluation}. In 
Section~\ref{sec:mptc_related_work} we present some related work on 
reconfigurable consensus and DoS attacks. Finally, we make our concluding 
remarks in Section~\ref{sec:mptc_conclusions}.
\end{comment}

%% file: model.tex
\section{Model}\label{sec:mptc_model}

\subsection{Processes and communication}
Our system consists of a set of processes $\mcn$ that 
communicate using message passing. Each process is modeled 
as a state machine with a potentially unbounded set of states that executes 
% deterministic or non-deterministic
transitions according to some protocol. The 
protocol specifies the transition function of the processes as well as the 
messages they exchange. Each process's state consists of two components, the 
\emph{public} and the \emph{private} or \emph{secret state}. The public state 
contains the description of the protocol that each process executes and any 
public cryptographic keys associated with the process. The secret state 
contains any run-time state the process manages during the execution of the 
protocol as well as any secret cryptographic keys and/or shares associated with 
the process. %We assume that all keys and shares are stored in a tamper-proof 
%cryptographic co-processor that executes all operations involving these keys. 

% REMOVED WEAK SYNCHRONY; NOW ASSUMPTION OF THE UNDERLYING PROTOCOLS
%Protocol execution and communication are \emph{weakly synchronous} meaning 
%that while there are bounds on the time it takes processes to execute 
%transitions and deliver messages, these bounds are not known to the 
%processes. This weak synchrony assumption is stronger than asynchrony but much 
%weaker than synchrony which assumes knowledge of the previous bounds by the 
%processes.
Protocol execution and communication are \emph{asynchronous}, meaning 
that there are no bounds on the time it takes processes to execute 
transitions and deliver messages.

A process can be correct or faulty. A correct process faithfully 
executes the protocol and is guaranteed to make progress as long as the 
conditions specified by the protocol at any given step are eventually met.
In this paper we will primarily consider crash failures, although we
extend our techniques to Byzantine failures in Appendix~\ref{sec:mptc_byzantine}.
A faulty process may crash at any time after which it stops executing the 
protocol. Up to the point of the crash, processes 
faithfully follow the steps of the protocol.
Communication between correct 
processes is reliable and secure. This means that,
in the absence of a DoS attack (see below), messages sent by some correct 
process to another correct process are eventually delivered. It also means that 
messages are authenticated and cannot be tampered with or fabricated.
% In this work we are not concerned with the recovery of failed processes.
We assume an 
upper bound, $f_c < |\mcn|$, on the total number of processes that might fail 
during the protocol's execution.

Finally, we assume the existence of a special process $T \notin \mcn$ that
from now on we will refer to as \emph{trusted dealer} or simply dealer. The 
dealer is only used during initialization of the system during which it 
generates the initial public and secret state of all processes. We assume that 
during this setup phase the dealer is correct, that it can communicate via 
secure channels with any process in the system, and that it does not disclose
its state. After initialization, however, the dealer does not execute any 
protocol steps or exchange messages with any other process, and the dealer's 
state is destroyed.

\subsection{Adversary and attacks}
\label{ssec:mptc_adversary}
We assume an adversary, $A$, that controls which processes fail and 
when. $A$ is limited on the number of processes it can fail by $f_c$ and cannot 
fail the dealer. $A$ can also control the delivery order of messages of all 
processes as well as delay communication, but must yield to the previously 
stated reliable communication assumption.

The adversary can also issue \emph{denial of service} (DoS) attacks against the 
system that can fully saturate the bandwidth resources of at most $f_a < 
|\mcn|$ correct processes. This can effectively prevent the targeted processes 
from progressing in the protocol's execution since they can no longer 
communicate with the rest of the system. $A$ can change the targets of an 
attack over time and, in this way, can introduce communication and computation 
delays on certain processes. The adversary's objective is to prevent the system 
from making progress. From now on we will denote by $f$ the maximum number of 
processes that can be crashed or under attack during the execution of the 
protocol, that is $f = f_c + f_a < |\mcn|$.

In this work, we ignore DoS attacks that target other 
resources like CPU using legitimate traffic. These attacks can be mitigated using rate limiting 
techniques such as client cryptographic puzzles~\cite{ANL01}.

The adversary has read access to the public state of all processes as well as 
the secret state of up to $f$ processes. We call the processes whose secret 
state is disclosed to $A$ \emph{compromised}. While $A$ cannot modify this 
state, it can use this state to select the target processes of a DoS attack. 
Once $A$ has selected the set of compromised processes it can no longer change 
that set thus preventing $A$ from accessing the secret state of more than $f$ 
processes. Note that the set of compromised processes is not necessarily 
related to the set of processes that are crashed or under attack.
% which would violate the security guarantees we rely upon in our protocol. 
%%%% FOR REASONING ABOUT DYNAMIC CORRUPTION %%%%
\begin{comment}
This static approach to 
compromising processes might seem restrictive and unrealistic. We 
can make our model more realistic by adopting the approach in~\cite{Zhou01} 
where the adversary can change the set of compromised processes over time. This 
requires defining a window of vulnerability during which the adversary can 
compromise up to $f$ processes but can change its selection across different 
windows. We can then use proactive refreshing of cryptographic keys that runs 
periodically to recover compromised processes across different windows. We will 
omit this discussion however, since this is not the focus of this work.
\end{comment}

Finally, we assume that the various cryptographic schemes we are 
employing, like public key cryptography and threshold signatures, are secure in 
the random oracle model. 
%This requires any computations performed by the adversary and the rest of the 
%processes to be polynomially bounded in some security parameter related to our
%cryptographic schemes.

\subsection{Cryptographic primitives}
\label{ssec:crypto_prims}
Our protocol relies on Threshold Coin Tossing~\cite{Cachin05}.
Here we present a high-level description 
of this primitive that we will further formalize in Section~\ref{sec:mptc_protocol}.
% \subsubsection{Threshold coin-tossing}
We employ an $(n, f+1, f)$ threshold coin-tossing scheme 
in which $n$ parties maintain shares of an unpredictable function, $F$, mapping 
an arbitrary bit string, $r$, to a binary value $\{0, 1\}$. Each of these 
shares can be used along with an input $r$ to create a value that from now on 
we will refer to as \emph{function shares}\footnote{The term used in 
\cite{Cachin05} for these values is coin shares}. At least $f+1$ of these 
function shares of $r$ are required to reconstruct the result $F(r)$, while at 
most $f$ parties may be compromised. We will use the term \emph{function share 
of} $F(r)$ to denote a function share of $r$ generated with a secret share of 
$F$.

The scheme defines three functions: 1) The \emph{split} function, which takes 
as input a function $F$ (represented as a bit string) and creates a set of 
shares as well as a verification key for each of these shares. 2) The share 
combining function, \emph{combine}, which takes an input $r$ of $F$ along with 
$f+1$ valid function shares of $r$ and produces $F(r)$. 3) The share 
verification function \emph{verify}, which takes an input $r$ of $F$, a 
function share on $r$, and the verification key corresponding to the share that 
generated the input function share and determines whether the function share is 
valid.

This scheme is based on threshold signatures~\cite{Shoup00} and can be used to 
create an unpredictable sequence of bits while ensuring that it is 
computationally infeasible for the adversary to produce an input $r$ and $f+1$ 
valid function shares that once combined do not yield $F(r)$. More formally, 
the scheme satisfies the following properties taken from~\cite{Cachin05}:

\begin{itemize}
	\item \emph{Robustness}: It is computationally infeasible for the adversary 
	to produce a value $r$ and $f+1$ valid shares of $r$ such that the result 
	of the combine function is not $F(r)$.
	\item \emph{Unpredictability}: Given a value $r$ and functions shares from 
	fewer than $f+1 - f$ correct processes, the adversary can predict the value 
	of $F(r)$ with probability at most $\frac{1}{2}+\epsilon$ for negligible 
	value $\epsilon \in \mathbb{R}$.
\end{itemize}

The previous unpredictability property was extended to sequences of output bits 
in~\cite{Cachin05}, such that, given a sequence of values $C_i$ for $i \in 
\{1,2, \ldots, b\}$, an adversary with fewer than $f+1 - f$ valid shares of 
some $C_i$ has negligible advantage in predicting $F(C_i)$. From now on, when 
we talk about unpredictability we will refer to this \emph{extended 
unpredictability property} of threshold coin-tossing.

Note that the previously described extended unpredictability property allows us 
to share unpredictable functions in $[\{0,1\}^* \rightarrow \{0,1\}^b]$ for any 
finite $b$. In other words, we can model each such function as a random number 
generator that can produce $2^b$ different values and requires $f+1$ processes 
to collaborate in order to produce the random (unpredictable) value 
corresponding to some arbitrary bit string $r$.
\begin{comment}
Also note that from now on, whenever we refer to sharing of functions or shares of functions we mean the unpredictable functions that we described in this section. 
\end{comment}

Threshold coin-tossing can be implemented using any non-interactive threshold signatures scheme that ensures unique valid signature per message as in~\cite{Shoup00}. A direct implementation of this scheme can be found in~\cite{Cachin05}.

\subsection{Underlying consensus protocols}
MPTC, like other consensus protocols, solves the problem of agreement.
In this problem, a set of possibly distributed processes, each of which is 
initialized with some input value, unanimously and irrevocably output one of 
those input values. More formally, let $\mcn$ be a set of processes each of 
which is initialized with some value from a value set $\mcv$. Each process can 
employ either of the following primitives: 

\begin{itemize}
	\item \emph{propose} a value which allows a process to communicate its 
	value to the rest of the processes in $\mcn$,
	\item \emph{decide} a value which allows a process to output a value.
\end{itemize}

Every correct consensus protocol must satisfy the following properties:

\begin{itemize}
	\item \emph{Validity}: If a process decides a value, then that value must 
	be the input value of some process in $\mcn$.
	\item \emph{Agreement}: If any two processes decide they must decide the 
	same value.
	\item \emph{Termination}: All correct processes eventually decide.
\end{itemize}

%In our protocol each consensus round is run by a subset of the system's 
%processes. It is possible that in an infinite execution some correct processes 
%may execute a finite number of consensus protocols rounds. As a result the 
%previous termination properties may be impossible to satisfy since only those 
%processes running sufficiently many rounds may be able to decide. To simplify 
%the correctness discussion of MPTC we will slightly modify the termination 
%property by employing Lamport's description of consensus~\cite{Lamport01} in 
%which he only required that the decision (chosen value) can be eventually 
%learnt by all correct processes. Note that to achieve this there must be at 
%least one correct process that eventually decides. The rest of the correct 
%processes can then trivially learn this decision by having the decided process 
%send its decision to the remaining correct processes in $\mcn$ and by relying 
%on reliable communication. We change the termination property as follows:
%
%\begin{itemize}
%	\item \emph{Termination}: At least one correct process eventually 
%	decides.
%\end{itemize}
%
%Note that the previous two definitions are in fact equivalent. If all correct 
%processes decide then obviously at least one must have and by the previous 
%argument if one correct process decides all correct processes can eventually do 
%so as well.

\cite{FLP85} has shown that in an asynchronous environment no consensus 
protocol exists that satisfies all of the above properties when even only a 
single failure can occur. To circumvent this result, a variety 
of protocols have been proposed \cite{BenOr83,CIL94} that use a 
probabilistic approach and can guarantee the previous properties with the 
following modification on termination: \emph{All correct processes eventually 
decide with probability 1}. 
%With our previous alternative definition termination can be stated as follows: \emph{At least one correct process eventually decides with probability 1}. 
For the remainder of this work we will refer to the non-probabilistic 
description of termination as \emph{definite termination} and to the 
probabilistic one as \emph{probabilistic termination}.

A consensus protocol that implements the previous specification (using either 
definite or probabilistic termination) even under the presence of $f$ crash 
failures is called $f$-crash-resilient. Note that our adversary can 
additionally perform denial-of-service attacks which can fully saturate a 
bounded number of processes and render them entirely unavailable. In an 
asynchronous environment there is no difference between a crashed process and a 
process that is under DoS attack from the other processes' perspective. For 
this reason we say that a consensus protocol is correct in our model if it is 
$f$-crash-resilient where $f=f_c + f_a$. From now on we will refer to such 
consensus protocols as $f$-resilient protocols.

%% The discussion below is not necessary since we can assume a computationally 
%% unbounded adversary due to our secure communication assumption
%It is worth noting that the previous termination criterion poses an issue when
%combined with the cryptographic framework we use in this work. The problem 
%stems from the fact that most cryptographic schemes are secure under the 
%assumption of a computationally bounded adversary. Most formalisms of 
%termination however involve infinite executions of the corresponding consensus 
%protocol which does not work under the computationally bounded assumption. We 
%can deal with this issue by employing the secret share refreshing protocols 
%from~\cite{Zhou01} with appropriately large window of vulnerability in order to
%prevent the adversary from eventually divulging secret information of non-
%compromised processes.

Each process executing MPTC may run different consensus protocols at different 
rounds. We denote the set of possible protocols each process can choose from by 
$\mcp$. Different consensus protocols make different assumptions under which 
they meet the previously described specification. The crash-tolerant consensus 
protocol of Ben-Or~\cite{BenOr83}, for instance, assumes an asynchronous 
environment and that each infinite schedule has a bounded number of processes 
performing a finite number of steps. Other protocols make assumptions such as 
bounds on the number of failures, different degrees of synchrony, the existence 
of failure detectors~\cite{CT96}, etc. We consider a consensus protocol correct 
if it satisfies agreement, validity and either definite or probabilistic 
termination. For each protocol $P \in \mcp$, we denote the set of assumptions 
required to hold for $P$ to be correct by $\mca_P$. In other words, if 
assumptions $\mca_P$ hold, then $P$ satisfies validity, agreement, and 
termination. A protocol $P$ is a valid candidate for $\mcp$ if it is correct 
under both the assumptions in $\mca_P$ and our previous model assumptions 
regarding failures, network reliability, and adversary.

We only consider consensus protocols operating in rounds and we follow the 
framework introduced in~\cite{NvR15} for the specification of the round 
outcomes. According to this specification, every process running a round of a 
consensus protocol ends up in one of the following states:
$\{D, U, M\} \times \mcv$,
where states $(D, v), v \in \mcv$ indicate that the process has decided $v$, 
states $(U, v), v \in \mcv$ indicate that no process has decided up to the 
current round, and finally, states $(M, v), v \in \mcv$ indicate that while the 
process is not decided, if a decision was made by some process then it must 
have been $v$. We will refer to these states as round outcomes or simply 
\emph{outcomes}. We denote by $o_p^r$ the outcome of process $p \in \mcn$ at 
the end of round $r \in \mathbb{N}$. 

More formally the following invariants hold about the outcomes of processes 
completing a round of a correct consensus protocol in $\mcp$:

\begin{invariant}
\label{inv:decision}
If $\exists p \in \mcn,\ r \in \mathbb{N}$ such that $o_p^r = (D, v)$, where $v 
\in \mcv$, then for each correct $q \neq p \in \mcn$ it holds that $o_q^r = (M, 
v)$ or $o_q^r = (D, v)$.
\end{invariant}

\begin{invariant}
\label{inv:undecided}
If $\exists p \in \mcn$, $r \in \mathbb{N}$ such that $o_q^r = (U, v)$ for some 
$v \in \mcv$ then $\forall q \in \mcn$, $u \in \mcv$: $o_q^r \neq (D, u)$.
\end{invariant}

This framework facilitates the description of MPTC in the next section and
can be used to describe most consensus protocols in literature, including 
\cite{BenOr83,CT96,Lamport98}.

\subparagraph{Problem}
Our goal is to design a round-based consensus protocol that is correct under 
the previous system and adversary assumptions and that runs a different 
existing consensus protocol on a different set of processes each round. The 
selection of protocols and processes for each round must not be predictable by 
the adversary without the collaboration of correct processes. For the purposes 
of this work, unpredictability is as described in 
Section~\ref{ssec:crypto_prims}.

%% file: algorithms.tex
\section{Moving Participants Turtle Consensus}
\label{sec:mptc_protocol}

In this section we describe our Moving Participants Turtle Consensus (MPTC) 
protocol. MPTC is an $f$-resilient consensus protocol operating in rounds such 
that in each round a different subset of processes may run a different 
consensus protocol. We start with some preliminary definitions and notation
and then describe the protocol.
% and finally sketch its correctness.

\subsection{Participants and participant sets}
MPTC is run by all processes in $\mcn$. In each round, only a subset of $\mcn$ 
is actively running a consensus protocol from a set of correct consensus 
protocols, $\mcp$. Let $\mcp_f$ correspond to the minimum number of processes 
required to run each protocol in $\mcp$. As an example, let $\mcp$ consist of 
the Ben-Or~\cite{BenOr83} and One-Third~\cite{CS09} consensus protocols. The 
first one requires $2f+1$ processes to solve the agreement problem tolerating 
up to $f$ crash failures while the second one needs $3f+1$. Thus $\mcp_f = 
3f+1$. We assume that $|\mcn| \gg f$ and thus $|\mcn| > \mcp_f$ for most 
$f$-resilient consensus protocols.
%Depending on the consensus protocol under execution a process might play 
%multiple roles e.g. in Synod of Paxos~\cite{Lamport98} a process might play the 
%role of an acceptor, a proposer or a learner. We will use Lamport's terminology 
%and refer to these roles as agents. Therefore, a different number of agents may 
%be distributed among the $\mcp_f$ processes running different protocols in 
%$\mcp$.

In the remainder of this paper, we say that a process \emph{runs} or 
\emph{executes} a protocol in $\mcp$ when it executes a round of that protocol. 
We will refer to a process executing a protocol in $\mcp$ at some round of MPTC 
as a \emph{participant} or an \emph{active participant} of that round. Let $PS 
= \{S \subseteq \mcn : |S| = \mcp_f \}$ be the set of all possible subsets 
of $\mcn$ where each subset has size $\mcp_f$. We call each such set a 
\emph{participant set}. A process may be a member of multiple participant sets. 
In each round $r$ of MPTC, only a single participant set, $S_r$, is 
\emph{active}, that is executing a consensus protocol in $\mcp$. We assume that 
participants in each participant set of some round $r$, $S_r \in PS$, are 
ordered and denote the $i$-th participant in $S_r$ as $S_r^i$. The active 
participant set for each round is determined by $T$ during initialization, 
which we describe later in this section.

\subsection{Configurations}
Before describing the initialization procedure and the core of MPTC, we need to 
define an important concept that encapsulates the information required for a 
set of processes to run a consensus protocol. We define a \emph{configuration} 
of MPTC as a tuple $(P, S) \in \mcp \times PS$. $P \in \mcp$ describes the 
consensus protocol to run along with its initialization parameters. To better 
understand the information contained in the initialization parameters, consider 
a protocol like Lamport's Paxos~\cite{Lamport98} and the core consensus 
protocol he called Synod. In Synod, processes play multiple roles, such as 
proposers and acceptors. In that sense, $P$ needs to encapsulate not only the 
protocol under execution, e.g. Synod, but also information related to its 
initialization such as mapping of proposers and acceptors to processes. The 
participant set $S \in PS$ corresponds to the set of processes that shall 
execute the consensus protocol specified by $P$. Let the set of all possible 
configurations $\mcc = \mcp \times PS$. Our approach implements a multi-party 
computation scheme for an unpredictable mapping between natural numbers 
(rounds) and configurations. We omit details regarding how to 
represent $P$ since this is an implementation issue and does not affect our 
protocol. We assume that $|\mcc|$ is bounded.

\subsection{Initialization and trusted dealer}
We are now ready to describe the initialization of our protocol, how we are 
using $T$ to create an unpredictable sequence of configurations, and how the 
active participants of a round can compute the corresponding configuration.

$T$ is a special process that generates the configuration that each process in 
$\mcn$ starts with in the first round. It also provides the processes the means 
to generate configurations for subsequent rounds. To achieve this, $T$ employs 
a $(\mcp_f, f+1, f)$ threshold coin tossing scheme like the one described in 
Section~\ref{ssec:crypto_prims}. Using this scheme, $T$ shares a function $F_S$ 
between the $\mcp_f$ processes of each participant set $S \in PS$. Recall 
that threshold coin-tossing can be implemented using threshold signatures, thus when we say that $T$ shares a function $F_S$ with each participant set, in reality it simply selects a different public-secret key pair for each $S \in PS$ and shares the secret key. Given some round number $r$, at least $f+1$ processes in $S$ need to collaborate to produce $F_S(r)$ while up to $f$ of them may get compromised. $f+1$ is both a sufficient and necessary number of processes to compute the result of the function shared. $T$ cannot be compromised, failed or attacked by the adversary.

At a high level, $T$ operates as follows:

\begin{enumerate}
    \item For each $S \in PS$ the dealer picks a function $F_S: \{0,1\}^* 
    \rightarrow \mcc$ and generates a secret share, $h_S^q$, for each $q \in 
    S$.
    \item $T$ picks a configuration $C_0 \in \mcc$.
    \item $T$ distributes $C_0$ and shares to processes over secure channels.
    $\forall S \in PS$ each process $p \in S$ receives $h_S^p$ and $C_0$.
\end{enumerate}

Observe that each function shared by the dealer maps arbitrary strings to 
configurations. This differs from the functions we defined in 
Section~\ref{ssec:crypto_prims} which map arbitrary bit strings to bit strings 
of some finite length $b$. Since $\mcc$ is finite, there exists $b = \lceil 
log_2|\mcc| \rceil$ such that we can trivially obtain an onto function 
$\{0,1\}^b \rightarrow \mcc$. Thus, the functions we need to share can be 
trivially obtained by the ones supported by the threshold coin-tossing scheme. 
Note that, by this high-level algorithm, a process in $\mcn$ will receive 
multiple shares, one for each participant set it belongs to. The dealer selects 
each $F_S$ such that the output is computationally indistinguishable from a 
randomly chosen function.

%We now discuss how $T$ generates the secret shares. Given model parameters 
%$\mcp$ and $f$, we define the \textit{split} function of the threshold 
%coin-tossing scheme presented in Section \ref{ssec:crypto_prims} as follows:
%\begin{equation}
%\textit{split}: [\{0,1\}^* \rightarrow \mcc] \rightarrow \mcs^{\mcp_f} 
%\nonumber
%\end{equation}
%
%\noindent where $[\{0,1\}^* \rightarrow \mcc]$ is the space of all functions 
%from arbitrary bit strings to configurations, and $\mcs$ is the space of secret 
%shares. \textit{split} breaks the input function into $\mcp_f$ shares one for 
%each process of some participant set. In other words, $\forall S \in PS$, 
%\textit{split}$(F_S) = \{h_S^p\ |\ p \in S\}$.
We now discuss how $T$ generates the secret shares. Given model parameters $
\mcp$ and $f$, $T$ generates a different set of secret key shares for each 
subset, $S \in PS$. Each such set of secret key shares implicitly defines a 
function $F_S$ mapping bit strings to configurations. We call this operation
\textit{split} and it is similar to the threshold coin-tossing dealer 
initialization described in~\cite{Cachin05}. \textit{split} can be implemented 
using Shamir's secret sharing~\cite{Shamir79} $(\mcp_f, f+1)$. Note that 
the implementation in~\cite{Cachin05} is based on verifiable secret sharing 
because they are considering Byzantine failures. In our model, processes cannot 
lie and messages cannot be tampered with. As a result, no verification is 
needed within this context. We extend our approach to byzantine failures in 
Appendix~\ref{sec:mptc_byzantine} where we introduce the required verification 
functionality.

Given a secret share, $h_S^p$, of some function $F_S: \mathbb{N} \rightarrow 
\mcc$ and some input, $r \in \mathbb{N}$, process $p \in S$ can create a 
function share of $F_S(r)$ using the \emph{share generation} function, $GFS: 
\mcs \times \mathbb{N} \rightarrow \mcf$ where $\mcf$ is the space of valid 
function shares that can be generated given a share $h \in \mcs$ and a natural 
number. We define, $F_S^p(r) = GFS(h^p_S, r)$. A straightforward implementation 
of $GFS$ can be derived from the signature share generation for threshold 
signatures~\cite{Shoup00}.

We define the \textit{combine} functions as:
\begin{equation}
\textit{combine}: \mcf^{f+1} \times \mathbb{N} \rightarrow \mcc \nonumber
\end{equation}

\noindent \textit{combine} works by receiving function shares of some function $F_S$ and some input, $r$ and outputting $F_S(r)$ which corresponds to a configuration. More formally, let 
%$F_S^Q(r) = \{F_S^q(r) \in \mcf | \forall q \in Q, Q \subseteq S \textit{ and } |Q| = f+1 \} \in \mcf^{f+1}$ be any set of $f+1$ function shares of $F_S(r)$. Then, \textit{combine}$(F_S^Q(r)) = F_S(r)$.
\begin{equation}
F_S^Q(r) = \{F_S^q(r) \in \mcf\ |\ \forall q \in Q, Q \subseteq S \textit{ and 
} |Q| = f+1 \} \nonumber
\end{equation}

\noindent be any set of $f+1$ function shares of $F_S(r)$, i.e. $F_S^Q(r) \in 
\mcf^{f+1}$. Then we have that:
\begin{equation}
\textit{combine}(F_S^Q(r), r) = F_S(r) \nonumber
\end{equation}

\subsection{Protocol description}
\label{ssec:prt_dsc}

We can now describe the operation of each process executing MPTC. MPTC is an 
$f$-resilient round-based consensus protocol in which each round is executed 
under a different configuration. Let $C_r = (S_r, P_r) \in \mcc$ be the 
configuration used for round $r$, where $S_r \in PS$ is the active participant 
set and $P_r \in \mcp$ the consensus protocol specification for that round. Let 
$o_r^p$ be the outcome of a process $p \in S_r$ running $P_r$ at round $r$. Let 
a special value $\bot \notin \mcv \cup \mcc \cup \{\{D, M, U\} \times \mcv\}$ 
represent the value of an uninitialized variable.

We assume that all processes have common knowledge of $\mcn$, $f$, $\mcp$, 
$\mcc$ as well as of the functions \textit{GFS} \textit{combine}. Each process 
$p \in \mcn$ runs MPTC with its identifier and some value $x_p \in \mcv$ as 
input and at any point in time maintains the following state:

\begin{itemize}
    \item its current round number, $r_p$, initialized to 0;
    \item its proposal \textit{proposal}$_p$, initialized to $x_p$;
    \item the outcome of a round, $o_p$ representing $p$'s decision state and 
    initialized to $\bot$ at the beginning of each round; and
    \item the current configuration $c_p$ describing the currently known active 
    participant set and the consensus protocol the active participants execute; 
    it is initialized to $C_0$, which is provided by $T$ during the 
    initialization phase.
    \item the secret shares, $h_S^p$, $\forall S \in PS$ such that $p \in S$ 
    provided by $T$ during initialization.
\end{itemize}

We have organized MPTC description in phases. Messages exchanged between 
processes carry the number of the phase, the id of the sending process, and the 
current round along with the payload. Messages are of the form $\langle 
\textit{phase number}$, $\textit{process id}$, $\textit{round}$, $\ldots 
\rangle $. Each round, $r$, of MPTC works in the following 3 phases:

\begin{itemize}
    \item \textbf{Phase 1}: Each process $p \in S_r$ runs a round of the 
    consensus protocol specified by $C_r$. 
    Let $o_p$ be $p$'s outcome for round $r$. If $o_p = (D, v)$, then process 
    $p$ updates $\textit{proposal}_p=v$, decides $v$ and never updates $o_p$ 
    and $\textit{proposal}_p$ again in any future round. If $o_p = (M, v)$, 
    then $p$ updates $\textit{proposal}_p=v$. Regardless of $o_p$'s value, $p$ 
    goes to Phase 2.
    \item \textbf{Phase 2}: 
    \begin{itemize}
        \item \textit{Step 1}: Each $p \in S_r$ computes function share 
        $F_{S_r}^p(r) = GFS(h_{S_r}^p, r)$ and sends a Phase 2 message 
        $\langle 2, p, r_p, o_r^p$, $F_S^p(r) \rangle$ to all processes in 
        $S_r$. Then $p$ waits for Phase 2 messages from $\mcp_f - f$ processes 
        in $S_r$. Once $p$ receives enough messages from some $Q \subseteq 
        S_r$, it proceeds to Step 2.
        \item \textit{Step 2}: If $o_p = (U, *)$ where $*$ can be any value in 
        $\mcv$, then $p$ updates its proposal to a value $v$, selected 
        arbitrarily from the outcomes contained in the received Phase 2 
        messages. It also updates $o_p = (U, v)$.
%        \item \textit{Step 2}: Once $p$ receives Phase 2 messages from a set of 
%        $\mcp_f - f$ processes, $Q \subseteq S_r$ and if $o_p = (U, *)$, then 
%        it updates its outcome and proposal as described below. Let $R$ denote 
%        the set of outcomes received:
%        \begin{itemize}
%	        \item Case 1: If $\exists o \in R$ such that $o = (D, v)$, then 
%	        process $p$ updates $\textit{proposal}_p=v$, decides $v$, sets its 
%	        outcome $o_p = (D, v)$ and never updates $o_p$ and 
%	        $\textit{proposal}_p$ again in any future round.
%	        \item Case 2: If $\forall o \in R$ it holds $o = (M, v)$ for some 
%	        $v \in \mcv$ then $o_p = (M, v)$ and $\textit{proposal}_p = v$.
%	        \item Case 3: If $(U, *) \in R$ where $*$ can be any value in 
%	        $\mcv$ or $\exists o, o^\prime \in R$ and $v, v^\prime \in \mcv$ 
%	        such that $(o = (M, v) \wedge o^\prime = (M, v^\prime)) \wedge (v 
%	        \neq v^\prime)$, then $p$ selects an arbitrary outcome $(*, v) \in 
%	        R$ where $*$ can be any value in $\{M, U\}$ and updates $o_p = (U, 
%	        v)$ and $\textit{proposal}_p = v$.
%        \end{itemize}
        \item \textit{Step 3}: Let $F_S^Q(r)$ be the set of function shares 
        received from processes in $Q$. $p$ computes the configuration of the 
        next round, $r+1$, as $C_{r+1} = \textit{combine} (F_{S_r}^Q(r), r)$ 
        and moves on to Phase 3.
    \end{itemize}
    \item \textbf{Phase 3}: Each $p \in S_r$ sends a Phase 3 message $\langle 
    3, p, r_p, o_p, C_{r+1} \rangle$ to each process in $S_{r+1}$. $p$ updates 
    its state: $r_p = r + 1$, $c_p = C_{r+1}$ and if it is still undecided, it 
    updates its outcome, $o_p = \bot$. Each process $q \in S_{r+1}$ that 
    receives Phase 3 messages with the same configuration value, $C_{r+1}$, 
    from $\mcp_f - f$ processes, updates its proposal as follows. Let $R$ 
    denote the set of outcomes received:
        \begin{itemize}
	        \item Case 1: If $\exists o \in R$ such that $o = (D, v)$, then 
	        $q$ updates $\textit{proposal}_q=v$, decides $v$, sets its outcome 
	        $o_q = (D, v)$ and never updates $o_q$ and $\textit{proposal}_q$ 
	        again in any future round.
	        \item Case 2: If $\forall o \in R$ it holds $o = (M, v)$ for some 
	        $v \in \mcv$ then $\textit{proposal}_q = v$.
	        \item Case 3: Otherwise, $q$ selects an arbitrary outcome $(*, 
	        v) \in R$ where $*$ can be any value in $\{M, U\}$ and updates 
	        $\textit{proposal}_q = v$.
        \end{itemize}
    Then $q$ sets $r_q = r+1$, $c_q = C_{r+1}$ and if it is 
    still undecided, it sets $o_q = \bot$. Finally, it starts the next round.
\end{itemize}

\begin{wrapfigure}{r}{0.5\textwidth}
  \vspace{-30pt}
  \begin{center}
    \includegraphics[width=0.5\textwidth]{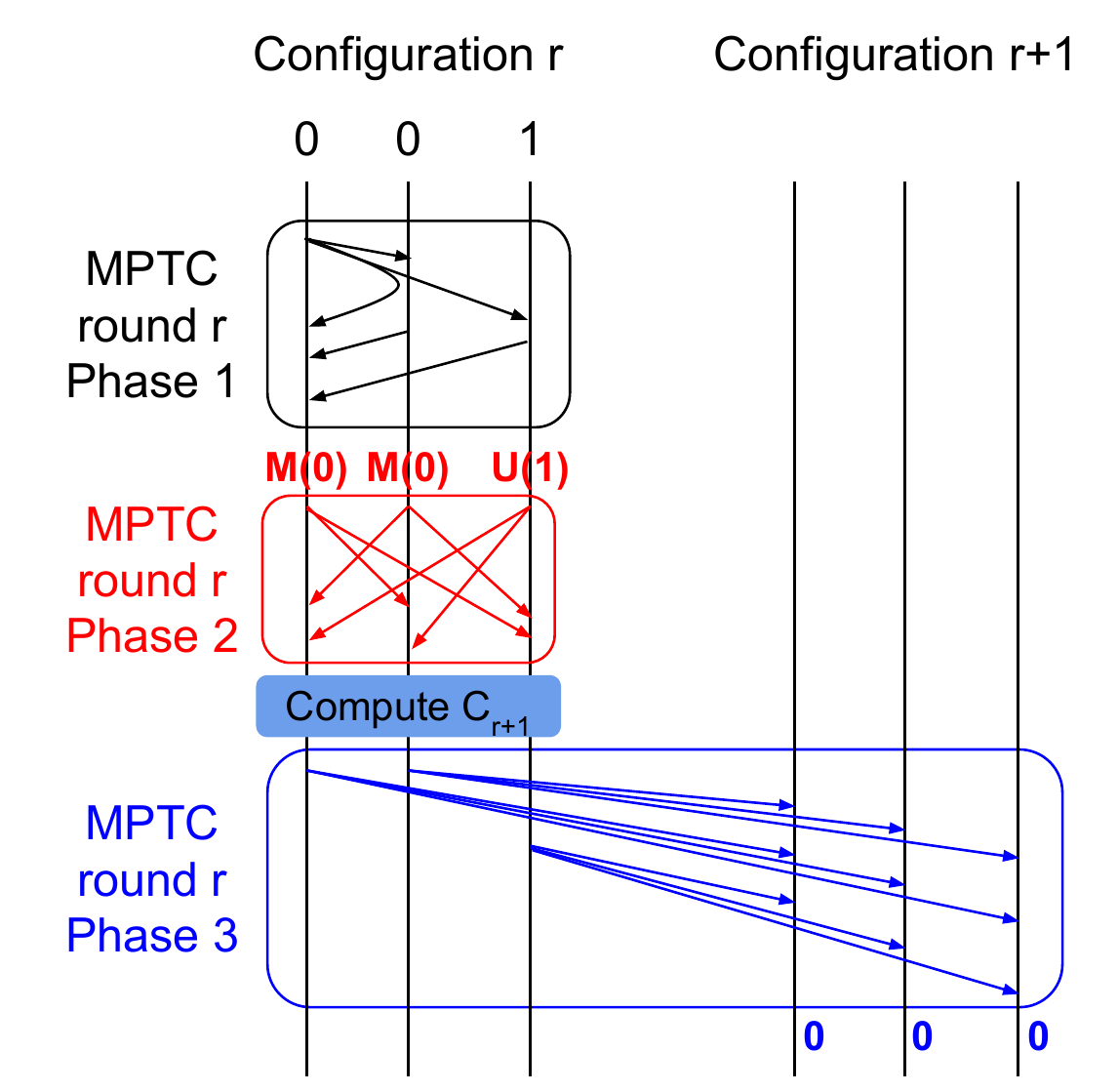}
  \end{center}
  \caption{A round of MPTC consensus.}
  \label{fig:mptc_protocol}
  \vspace{-35pt}
\end{wrapfigure}

Figure~\ref{fig:mptc_protocol} shows a visualization of the previous round 
description. MPTC runs for an unbounded number of rounds and eventually reaches 
a state in which a decision is made and all correct processes can eventually 
learn this decision. Messages from old rounds, either delayed in the network or 
sent by slow processes, are ignored while messages from future rounds are 
queued to be processed when the receiver reaches that round. The correctness of 
the previous protocol is presented in Appendix~\ref{ssec:mptc_proofs}. In 
addition, in Appendix~\ref{sec:mptc_byzantine} we present a byzantine-tolerant 
version of the previous protocol along with its own correctness discussion.

%% file: implementation.tex
\section{Implementation}\label{sec:mptc_implementation}

In this section, we describe a simple implementation of a non-byzantine version 
of MPTC as well as a state machine replication protocol we built on top of it. 
To implement MPTC we need to decide on the following parameters: the choice of 
protocol set $\mcp$, the set of possible configurations $\mcc$, the 
configuration selection functions $F_S$, $\forall S \in PS$, generated by the 
trusted dealer, and the implementations of \textit{split}, \textit{GFS} and 
\textit{combine} functions.

Our set of protocols, $\mcp$, contains only a single consensus protocol, a parameterized version of single decree Paxos~\cite{Lamport98} in which each round comes with a predetermined leader known to all active participants. Paxos tolerates $f$ crash failures using $2f+1$ processes and under failure-free execution conditions, it can reach a decision within a single round-trip of communication. We assume the weakest failure detector, $\diamond \mcw$, presented in~\cite{CHT96} which we implement using timeouts with exponentially increasing timeout periods. This way we ensure that there will be enough rounds executed
%``concurrently''
by sufficiently many processes, which is critical for ensuring termination in our Paxos variant.

The timeouts mentioned above may cause certain processes executing our Paxos variant to exit a round without knowledge of the round's decision. Such processes need to retrieve this knowledge from the rest of the processes. To avoid incurring another round of communication in our Paxos variant, we piggyback this decision state retrieval onto Phase 2 of MPTC. Timed out processes can use the set of outcomes received to update their proposal.

Our set of configurations is $\mcc = \{(S, P)\ |\ S \in PS \text{ and } |S| = 
2f + 1\}$ where $P \in \mcp$ is the described Paxos variant. Observe that in 
contrast to prior work on Turtle Consensus~\cite{NvR15} we use the same 
protocol across configurations. In Turtle Consensus, different configurations 
used the same $2f+1$ set of processes. As a result, the adversary could try to 
track the current leader within that set of processes even if the leader 
changed across different configurations. Therefore, a competent adversary could 
eventually locate and force Turtle Consensus rounds to fail, which can lead to 
poor performance. For that reason, Turtle Consensus kept switching between a 
leader-based (Paxos) and fully decentralized (Ben-Or) consensus protocols 
across configurations to prevent the adversary from exploiting the leader 
vulnerability. A side-effect of that approach, however, was that by falling 
back to a less efficient protocol (Ben-Or) it only achieved sub-par performance 
compared to the graceful execution using only Paxos rounds. With MPTC we do not 
need to employ such tactics since the adversary now needs to scan through $|
\mcn| \gg f$ processes before it can identify the leader of our Paxos 
configuration.

In the implementation that we evaluate in Section~\ref{sec:mptc_evaluation} we 
did not implement the Threshold coin-tossing scheme. We emulated it instead by 
assuming that all participant sets use the same unpredictable function given to 
all processes via a configuration file. This file defines a sequence of 
configurations, one for each round, that processes move to in a round-robin 
fashion. We emulate the restrictions that the cryptographic framework imposes 
on the adversary by assuming that only the processes involved in rounds $r$ and 
$r+1$ can learn $C_{r+1}$ and only after Phase 2 of round $r$ completes.

The interested reader can find an actual implementation of Threshold coin-
tossing in~\cite{Cachin05}. In that work, they used cryptographically secure 
hash functions modeled as random oracles to implement unpredictable functions 
as well as for the \textit{GFS} function. They also used Feldman's verifiable 
secret sharing~\cite{Feldman87} for \textit{split} function, though in our non-
byzantine case Shamir's secret sharing~\cite{Shamir79} can be used instead. 
Finally, for \textit{combine} they use Lagrange interpolation with coefficients 
the computed function shares.

For more details, see Appendix~\ref{sec:protocol_implementation}.

\subsection{MPTC-based state machine replication}
\label{ssec:mptc_smr}

We used the previous implementation of MPTC to build a SMRP, similar to the one 
described in~\cite{NvR15}. While the components of the 
implementation are similar, their interactions are different. There are three 
sets of processes, the clients, the replicas $\mcr$, and the participants 
$\mcn$. The clients issue requests to the participants who order these requests 
and forward them to replicas. Replicas execute the received requests in the 
order established by participants and send the results back to participants who 
then forward them back to clients. Participants can additionally send 
reconfiguration messages to each other in order to update the configuration of 
the MPTC execution.

In greater detail, clients send uniquely identifiable requests to sufficiently many participants in order to ensure that at least one correct participant receives each request. The participants receive requests from clients and are responsible for ordering these requests and send them for execution to the replicas. Only one participant set can be active at any point in time. Any participant outside that set receiving a client request relays that request to the currently known active participant set. Active participants receiving client requests spawn MPTC instances, one for each request that needs to be ordered. Each instance has its own identifier and decided requests are ordered according to the identifiers of the MPTC instances that decided them. Clients can only communicate with participants and thus they are unable to launch DoS attacks on the replicas. MPTC is lazily instantiated for each slot and MPTC messages carry instance identifiers so incoming protocol messages are properly processed by the correct instance. If an instance has not yet been created, messages for that instance are queued and processed when it is created. Finally, there are at least $f + 1$ replicas, each of which maintains a copy of state of the service implemented by the SMRP. All replicas are initialized in the same state and execute the clients' requests in the order determined by id of the consensus instance created by the participants for each request.

For a detailed description of this SMRP implementation see Appendix~\ref{sec:mptc_smr}.

%% file: evaluation.tex
\section{Evaluation}\label{sec:mptc_evaluation}

In this section we present an evaluation of MPTC using the SMRP protocol 
presented in Section~\ref{ssec:mptc_smr}. In Section~\ref{ssec:mptc_setup} we 
present the experiment setup and in Section~\ref{ssec:mptc_results} the 
performance results of MPTC under different attack scenarios.

\subsection{Setup}
\label{ssec:mptc_setup}

We implemented MPTC and the SMRP described in Section~\ref{ssec:mptc_smr} using 
\verb!C++!. Our testbed consists of 10 nodes in Emulab~\cite{emulab02}, each 
with 8 cores running at 2.4 GHz, with 64GB of memory. For our experiments we 
used $f=1$. Two nodes where designated as replicas, six as participants, one as clients, and one as the attacker. Nodes are connected by 1Gbps switched Ethernet as shown in Figure~\ref{img:mptc_topology}. Note that clients and attacker can only connect to participants, while participants connect to both replicas and clients. This choice was made to disable the attacker from directly attacking the replicas of SMRP and thus degrading performance without attacking the consensus mechanism. All communication between participants takes place through Switch 1 in our topology. Switch 2 is only used for participant to replica communication. We do not allow participants to communicate through Switch 2 since this would prevent the attacker from saturating the participants' bandwidth with respect to the MPTC execution. This would give MPTC an unfair advantage and would not showcase the benefits of its reconfiguration capabilities. All communication is over TCP/IP except for the DoS attack traffic, which is entirely UDP/IP. One of the two client nodes is used by the attacker and the other for creating legitimate client threads. We use a separate node for attacks in order to limit the effect of bandwidth attacks on the clients' ability to issue requests.

%\begin{figure}[!h]
%  \centering
%    \includegraphics[width=0.6\textwidth]{images/mptc_topology}
%  \caption{Experiment topology.}\label{img:mptc_topology}
%\end{figure}

\begin{wrapfigure}{r}{0.5\textwidth}
  \begin{center}
    \includegraphics[width=0.5\textwidth]{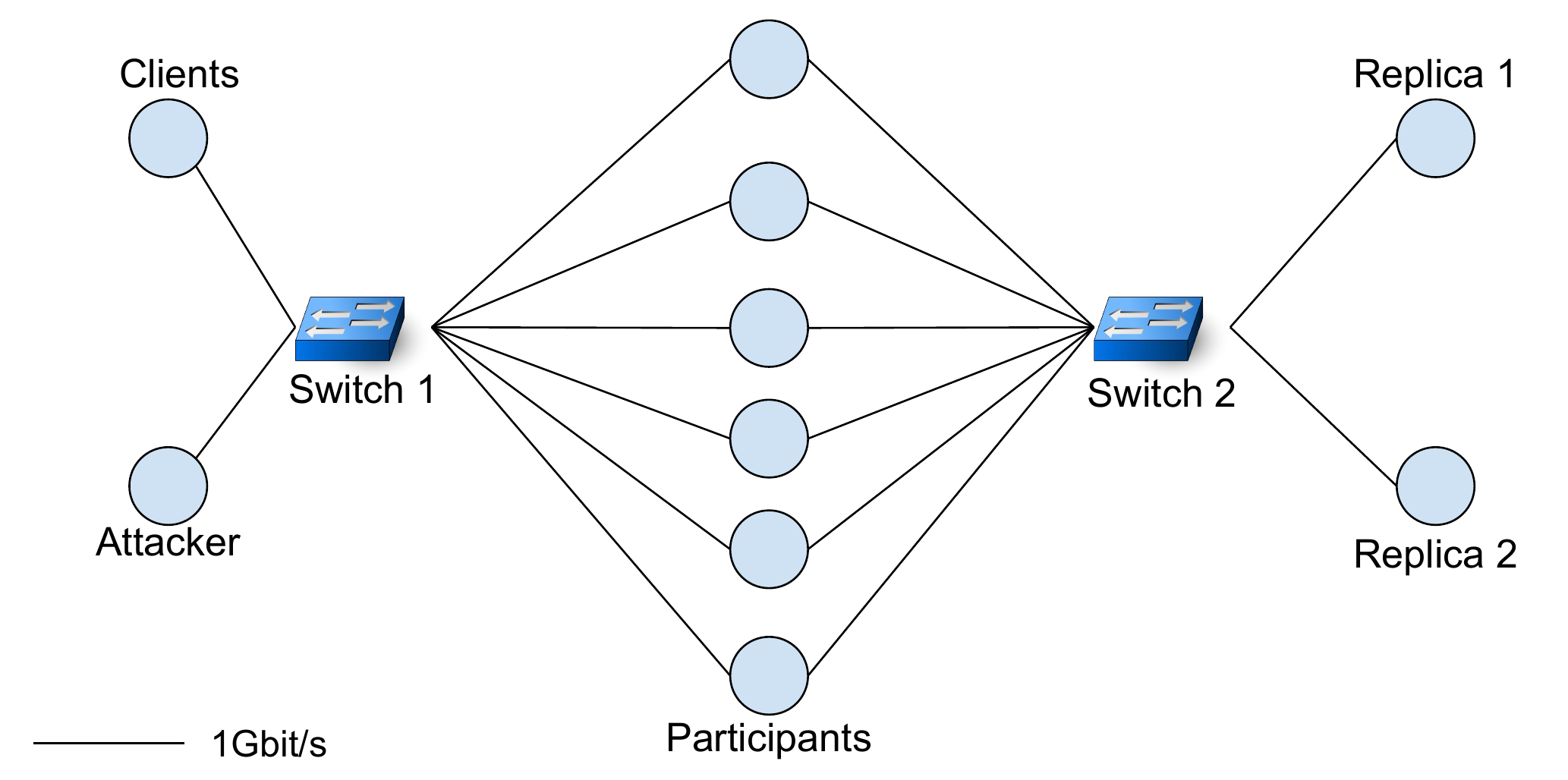}
  \end{center}
  \caption{Experiment topology.}
  \label{img:mptc_topology}
\end{wrapfigure}

To simplify our evaluation, we set $\mcc$ to contain only two configurations 
such that the corresponding participant sets are disjoint. The configuration 
selection function provided by the trusted dealer (in our implementation by a 
configuration file) simply alternates between these two configurations every 
time a round fails. The predetermined Paxos leader of each configuration 
depends on the round in which the configuration is run and is rotated in a 
round-robin fashion every time the same participant set is reused. We consider 
that the attacker does not have this knowledge to make informed decisions 
regarding targeting processes.

Clients first connect to $f+1$ random participants to which they issues 
requests. Once connected, each client executes the following loop: It issues 
each request to all $f+1$ participants, waits for a response, discarding 
duplicate responses, and then sends the next request. Note that by connecting 
to $f+1$ participants, we ensure that each client request reaches at least one 
correct participant who will further forward the request to the active 
participants. We have client requests contain no-ops, which means that when a 
decided request becomes ready for execution, replicas can immediately reply 
with a response.

The attacker creates a small number of attack threads, each of which targets a 
single participant, selects a random port, and sends UDP dummy messages as fast 
as it can. Note that these messages are not requests and are not processed by 
our participants since they never get to the application level. As in the 
Turtle Consensus evaluation~\cite{NvR15}, the goal of the attack is to prevent 
at most one participant from participating in MPTC instances. The attacker can 
focus all threads on the same participant or spread them across different ones. 
Since all attack threads are created on a single node, the aggregate bandwidth 
the attacker threads can saturate from the service cannot exceed 1Gbps.

We conducted experiments to test the throughput and latency of our 
implementation under normal execution and DoS attacks. Both metrics were 
measured at the client side. For throughput we measured the aggregate number of 
operations per second completed by client threads. Note that this is not the 
actual number of instances completed per second by our SMRP implementation 
since the same request might be decided more than once.

Other parameters of our experiment include: 

\begin{itemize}
    \item Duration: Each experiment lasted 1 minute. We found longer experiments did not significantly affect our metrics.
    \item Load: The number of concurrent clients, which ranged in our experiments from 1 to 64.
    \item Request size: The size of the command contained in each client request, which we set to 100 bytes.
    \item Attack message size: The size of the UDP messages send by attack threads to saturate the participants bandwidth; we set that to 1KB since our experimentation with our platform showed it is the smallest message size with the best results for the attacker.
    \item Number of attacker threads: Each run involving a DoS attack had 8 attack threads. We found that this number of threads yields best results for the attacker even when all 64 clients are connected to the target sharing the same link.
    \item Timeout: This is the initial timeout period used in our Paxos variant (Section~\ref{sec:mptc_implementation}) for each MPTC instance. Every time a round of some instance fails we double the timeout period for that instance.
\end{itemize}

\subsection{Results}
\label{ssec:mptc_results}

In our evaluation, we investigated three main scenarios. In the first, we run 
our implementation of MPTC without any attacks taking place. The performance of 
this scenario will be our baseline since any attack scenarios drain resources 
from the system and thus is expected to perform similar or worse. This scenario 
is labeled ``No attacks'' in our figures. 

The second scenario has the attacker focusing the DoS attack on a single node, 
the one that hosts the Paxos leader. This attack depletes the leader's 
bandwidth. In this scenario no reconfiguration occurs. More specifically, we 
assume that in each round of MPTC the exact same configuration is chosen and 
the leader remains the same. Note that this scenario tries to simulate the case 
where the adversary can accurately track and attack the leader of the Paxos 
configuration. While any reasonable implementation of Paxos would change 
leaders among the $2f+1$ processes, we set up the scenario to simplify issuing 
a very efficient attack. In our figures, this scenario is labeled ``Attack 
leader without reconfiguration''.

Finally, the third scenario uses an attacker who like in the previous scenario focuses on a single node. In this scenario the attacker is given the initial position of the leader but this time our implementation uses the MPTC version we described in Section~\ref{ssec:mptc_setup} where consensus instances execution alternates between two disjoint sets of nodes. The attacker strategy here is to saturate the bandwidth of the known leader. It keeps attacking that node for the entirety of the experiment run. This attack is labeled ``Attack leader with reconfiguration''.

%\begin{figure}
%  \centering
%    \includegraphics[width=1\textwidth]{figures/throughput}
%  \caption{Throughput as a function of load and various attack scenarios.}\label{fig:mptc_throughput}
%\end{figure}

%\begin{figure}
%    \centering
%    \begin{subfigure}[b]{0.48\textwidth}
%        \includegraphics[width=\textwidth]{figures/throughput}
%        \caption{Throughput as a function of load and various attack scenarios.}
%        \label{fig:mptc_throughput}
%    \end{subfigure}
%    ~ %add desired spacing between images, e. g. ~, \quad, \qquad, \hfill etc. 
%      %(or a blank line to force the subfigure onto a new line)
%    \begin{subfigure}[b]{0.48\textwidth}
%        \includegraphics[width=\textwidth]{figures/latency}
%        \caption{Latency as a function of load and various attack scenarios.}
%        \label{fig:mptc_latency}
%    \end{subfigure}
%    \caption{Moving Participants Turtle Consensus performance under different attack scenarios.}\label{fig:mptc_performance}
%\end{figure}

\begin{figure}[!ht]
	\captionsetup[subfigure]{justification=centering}
    \centering
    \subfloat[Throughput as a function of load.\label{fig:mptc_throughput}]{%
    	\includegraphics[width=0.48\textwidth]{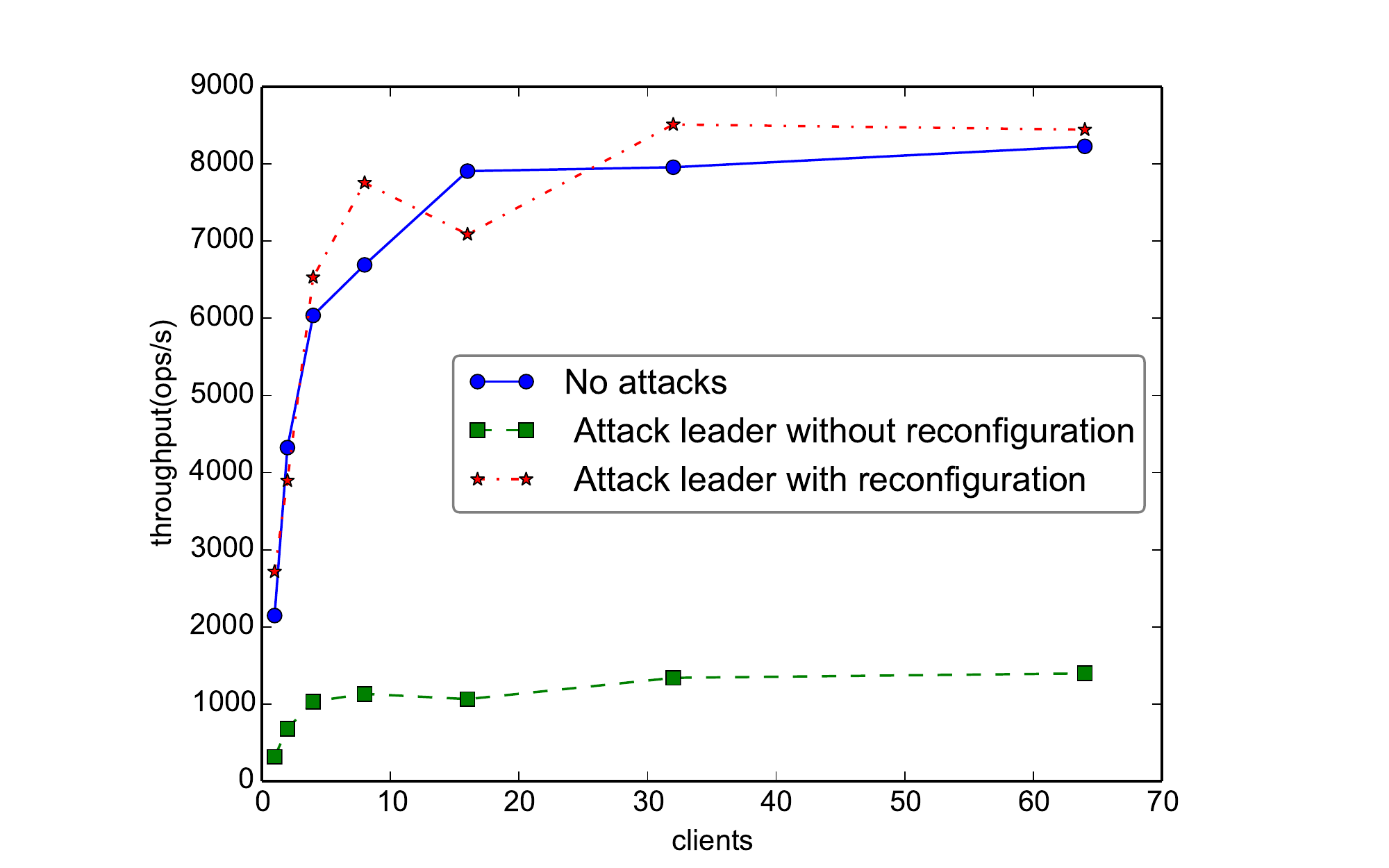}}
    \hfill
    \subfloat[Latency as a function of load.\label{fig:mptc_latency}]{%
    	\includegraphics[width=0.48\textwidth]{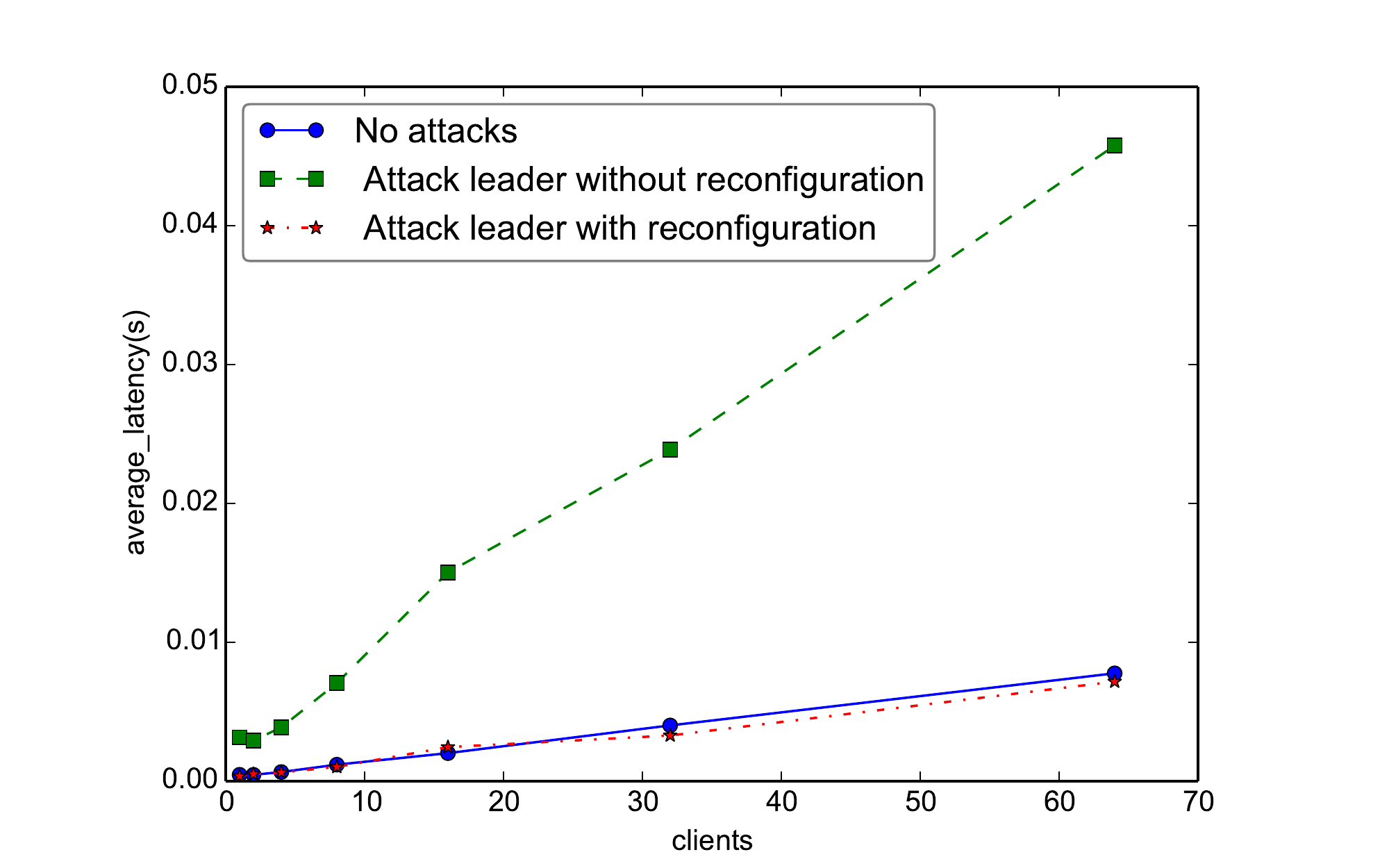}}
    \caption{Moving Participants Turtle Consensus performance under different attack scenarios.}\label{fig:mptc_performance}
    \vspace{-0.5cm}
\end{figure}

Figure~\ref{fig:mptc_throughput} shows the throughput comparison of the 
previous three experiment scenarios as a function of the load on the SMRP. Each 
point represents the average throughput over 10 runs for each number of 
clients. In each of these runs clients connect to random participants, which in 
turn means that performance will vary across experiments. The first scenario is 
our best case scenario since the system operates at full resource capacity. The 
second scenario shows that performance suffers substantially when the Paxos 
leader is under attack. This is to be expected since the leader's participation 
is critical for making progress in each MPTC instance. In the third scenario we 
observe the benefits of the reconfigurable version of MPTC in action. The SMRP 
throughput is close to that of the No Attacks case. The main reason for this 
behavior is that since the leader of the first configuration is under attack 
and lacks the bandwidth to handle the valid traffic, some instance will 
inevitably fail the first round since the remaining participants will 
eventually time out. That will cause a reconfiguration that changes the active 
participant set. The new participants will pick up the failed instances as well 
as future requests and continue operating at full capacity. The minor 
deviations observed between scenarios 1 and 3 are mainly due to the randomness 
of client distribution over the set of all participants.

%\begin{figure}
%  \centering
%    \includegraphics[width=1\textwidth]{figures/latency}
%  \caption{Latency as a function of load and various attack scenarios.}\label{fig:mptc_latency}
%\end{figure}

Figure~\ref{fig:mptc_latency} shows a comparison of the same scenarios as the load increases, but this time with respect to latency. Observe that all scenarios behave similarly with the latency linearly increasing with the load. This behavior is to be expected since, as the load increases, the number of concurrent MPTC instances increases, which in turn increases the latency for each client. After all, each of them has to wait for a response to their previous request before sending the next one. As in the case of throughput, we see that both scenarios 1 and 3 have similar latencies while scenario 2 performs poorly. The reasoning is the same. In the second scenario the leader under attack is slower in completing instances, which raises the wait time for each client.

Note that this evaluation does not take into account the
additional cost of reconfiguration that stems from the cryptographic
operations required for threshold coin tossing like RSA exponentiations.
We therefore expect that under frequent reconfigurations there will
be a wider gap between the performances of scenarios~1 and~3.
However, we also expect that such reconfigurations will be infrequent,
especially as the number of processes increases. Thus, while not
an absolute comparison, our evaluation showcases the expected
behavior and advantage of MPTC.

%% file: relatedwork.tex
\section{Related Work}\label{sec:mptc_related_work}

A wide range of crash-tolerant consensus protocols have been proposed in 
literature each optimized for a different setting and/or metric. Some were 
designed to handle datacenter-scale systems like~\cite{CGR07} which describes 
how Paxos was used to implement a fault-tolerant database for the Chubby 
locking service, an instance of which lies in each Google's datacenter. Others 
are focused on wide area deployments such as Mencius~\cite{MJM08}, which is a 
Paxos variant that employs multiple leaders each of which is responsible for a 
different set of consensus instances and may reside at different datacenters. 
Another important differentiating aspect of consensus protocols is whether they employ a special leader process like in~\cite{CT96,Lamport98} or whether they are fully decentralized like the protocol proposed in~\cite{BenOr83}.
This can greatly affect the behavior of a consensus protocol under different failure scenarios, including attacks, and was thus used by previous work on reconfigurable consensus~\cite{NvR15} to design consensus protocols that provide acceptable performance under certain DoS attacks.

%MPTC is an extension of the Turtle Consensus protocol~\cite{NvR15}. Turtle Consensus exploits these diverse characteristics of different consensus protocols by switching the protocol under execution during runtime. Our extension adds another degree of freedom in the Turtle Consensus reconfiguration capabilities by additionally altering the set of processes executing the consensus protocol. In addition, our configuration selection across rounds while still predetermined as in~\cite{NvR15} is now unpredictable thanks to the use of cryptographic techniques.

Our work resembles the work on Vertical Paxos~\cite{LMZ09}. Vertical Paxos is a reconfigurable state machine replication protocol that uses a special auxiliary master process to decide the next configuration of the system including the set of replicas participating in that configuration. Unlike Vertical Paxos, MPTC does not require additional master processes to be constantly active in order to compute the next configuration. Our assumed trusted dealer is only active during initialization. In addition, Vertical Paxos is not designed for an adversary capable of compromising even a single process and thus would not perform as well against the DoS attacks described in this work.

%\subsection{Denial-of-service attacks and moving target defense}
%\label{ssec:mtd}
%
%Denial-of-service (DoS) attacks are a common threat and have been studied in a 
%variety of distributed systems settings, from internet web 
%services~\cite{LRST00} to sensor networks~\cite{WS02} and smart power 
%grids~\cite{ML11}. The attacks considered in this work target the fault-
%tolerance component of a distributed system by attacking one of its core 
%building blocks, its consensus mechanism. Since such fault-tolerance mechanisms 
%are common in many distributed system we believe our method to be applicable 
%in most of the previous settings.

Moving target defenses have often been used as response to DoS and Distributed 
DoS (DDoS) attacks. \cite{GW00} proposes changing the IP address of the target 
node for dealing with local IP-based DoS attacks. More recently 
in~\cite{JAD12}, Software-Defined Networking (SDN) has been used to implement 
moving target defense approaches like ``random host mutation'' in which, 
similarly to~\cite{GW00}, the controller periodically alters the virtual IP 
addresses of hosts to hide the real IP addresses from an intruder. Our Moving 
Participants Turtle Consensus approach resembles more the ``proactive server 
roaming'' approach in~\cite{KSMMZ03}. That is an adaptive approach in which the 
active server proactively switches servers from an existing pool in order to 
deal with unpredictable and undetectable attacks. Their approach ensures that 
only legitimate clients can track the moving server. Like in the case of our 
MPTC protocol, proactive server roaming performs gracefully during attacks. 
However, it imposes significant overhead in attack-free scenarios, which is not 
the case for MPTC since we only reactively change configurations. %For a more in-depth description of the challenges, achievements, and future directions on the topic of moving target defense approaches see~\cite{JGSWW11}.

Our work assumes an adversary that cannot change the set of corrupted 
processes over time. Other related work has focused on dynamic models of 
corruption. \cite{HJKY95} introduced proactive secret sharing, an instance of 
proactive security~\cite{OY91} for supporting secure computation 
in synchronous distributed systems. These ideas have been adapted to 
asynchronous ones in~\cite{CKLS02,ZSV05}. While these approaches did not 
consider DoS attacks, they are orthogonal to ours and can be used to further 
improve this work for dealing with mobile adversaries.

Running consensus on a subset of a larger set of processes to decrease message 
complexity has been explored in~\cite{AAKS14}. It has also been explored more 
recently in~\cite{LNZBGS16} for improving the scalability of Byzantine 
agreement on blockchains.

%% file: conclusions.tex
\section{Conclusions}
\label{sec:mptc_conclusions}

In this paper we presented Moving Participants Turtle Consensus (MPTC), an 
extension to the Turtle Consensus protocol~\cite{NvR15} that allows running 
different consensus protocols, on different sets of processes, across different 
rounds of a single consensus instance. 
%By altering the execution configuration 
%of the consensus mechanism on the fly we believe we can better protect a system 
%under DoS attacks issued by adversaries that target configuration-specific 
%vulnerabilities. 
% We described MPTC for both crash and byzantine failure environments.
MPTC can deal with adversaries with bounded information on the 
system by making unpredictable changes in the execution of the protocol. 
%Thus, 
%our protocol eliminates adversary's advantage on predicting protocol and 
%execution-specific information that can be used against it. 
%We built a prototype implementation of MPTC in which we used the same protocol 
%across configurations and kept changing the set of processes executing the 
%different rounds. 
Our evaluation of our prototype implementation of MPTC suggests that we can achieve the performance offered by the most efficient consensus protocols even when the system is under attack. 
%Thus, MPTC improves on the core Turtle Consensus protocol not only because it can handle stronger and better informed adversaries but also due to its ability to maintain good performance without having to resolve to less efficient consensus strategies.

There are various directions for further exploration and improvement of MPTC. 
First, MPTC should be tested under more sophisticated attacks, for example
in which the attacker keeps changing the target of the attack even after 
reconfigurations, until it hits the node whose absence most impacts 
performance. In addition, we would like to design a dynamic configuration 
selection scheme in order to extend MPTC's applicability to different 
environments where unexpected changes in the workload of the system may lead to 
sub-optimal configurations.

%% file: correctness.tex
\section{Correctness of Moving Participants Turtle Consensus}
\label{ssec:mptc_proofs}

MPTC needs to satisfy the correctness properties of consensus protocols we 
defined in Section~\ref{sec:mptc_model}. We will show these properties assuming 
that all the underlying protocols used in $\mcp$ are $f$-resilient.

Recall that each protocol $P \in \mcp$ potentially makes additional assumptions 
regarding the system model under which the protocol satisfies the previous 
properties. To guarantee correct execution of $P$ under MPTC, the protocol must 
be correct under the previous set of assumptions and those assumptions made in 
Section~\ref{sec:mptc_model}. Note that this union may contain different 
assumptions about the same property or aspect of the system. Let $P$ be a 
consensus protocol for synchronous systems. Then this union of assumptions will 
contain our assumption that the system is asynchronous and $P$'s assumption 
that it is synchronous. In such cases, where one assumption is stronger than 
another with respect to a particular aspect of the model, we assume that the 
stronger assumption holds. In our example that would mean that the system would 
be synchronous.

To facilitate our discussion we will introduce some notation to describe 
collections of assumptions for the protocols in $\mcp$. Let $\mca$ denote the 
set of assumptions that we made in Section~\ref{sec:mptc_model}. For each $P 
\in \mcp$, let $\mca_P$ denote the set of assumptions $P$ makes about the 
system and its operation in order to tolerate $f$ crash failures while 
satisfying the correctness criteria stated in Section~\ref{sec:mptc_model}. 
Such assumptions may regard the minimum number of processes required as a 
function of $f$, the behavior of the network, the time and ordering of events, 
restrictions on the adversary, and many more. In this work we will not attempt 
to accurately model such assumptions since their descriptions vary greatly 
among different consensus protocols proposed in literature. 
%We will instead treat them as abstractions upon which we will consider different relationships.

Given $\mca$ and $\mcp$ we define the set of assumptions that must hold in every execution of MPTC as

\begin{equation}
\mca_\mcp = (\bigcup_{P \in \mcp} \mca_P) \bigcup \mca \nonumber
\end{equation}

\noindent In other words, every execution of MPTC must respect the union of all assumptions made by every individual protocol used in Phase 1 as well as our model's assumptions. If an assumption is overridden by another, the stronger of the two holds and thus weaker assumptions are excluded from $\mca_\mcp$. Under this definition of $\mca_\mcp$ we have that $\forall P \in \mcp$ every execution of $P$ under $\mca_\mcp$, respects the correctness properties of consensus protocols.

We will first discuss validity and agreement and finally we will argue about 
termination.

\subsection{Validity}
In order to satisfy validity, MPTC needs to ensure that any value decided must 
be the input value $x_p$ of some process $p \in \mcn$. This is encapsulated in 
the following lemma:

\begin{lemma}
\label{correctness:validity}
If a process in $\mcn$ running MPTC under $\mca_\mcp$, decides a value $v \in 
\mcv$, then $\exists p \in \mcn : x_p = v$.
\end{lemma}

\begin{proof}
\emph{(Sketch)} A process running MPTC can decide during Phase 1 or Phase 3. 
Any decision made during Phase 1 respects validity by correctness of each 
protocol in $\mcp$ under $\mca_\mcp$. By the outcome and proposal update 
function of Phases 2 and 3 we described in Section~\ref{ssec:prt_dsc}, the 
processes' proposals and thus possible decision values can only come from the 
values of the outcomes produced in Phase 1. Since validity holds during Phase 
1, if some value $v$ gets decided at some round $r$, $\exists p \in \mcn$ whose 
input value in $r$ is $v$.
\end{proof}

From the previous lemma, we have that MPTC initialized with a valid $\mcp$ 
satisfies validity under $\mca_\mcp$.

\subsection{Agreement}
Agreement is a safety property that ensures that no bad states occur during the system's operation. All correct consensus protocols must satisfy the agreement property for the entirety of their execution. Thus we need to show:

\begin{lemma}
\label{correctness:agreement}
If any two processes running MPTC under $\mca_\mcp$ decide values $v$ and 
$v^\prime$ respectively, then $v = v^\prime$.
\end{lemma}

\begin{proof}
\emph{(Sketch)} Assume that $\exists p, p^\prime \in \mcn$ that decide values 
$v$, $v^\prime \in \mcv$ in rounds $r, r^\prime \in \mathbb{N}$ respectively 
such that $v \neq v^\prime$. We have the following two cases:
\begin{itemize}
    \item \emph{Case} $r = r^\prime$: For two different processes to decide 
    different values it must be the case that outcomes $(D, v)$, $(D, 
    v^\prime)$ are computed after Phase 1 of $r$. This cannot occur because 
    each $P \in \mcp$ is correct under $\mca_\mcp$ and therefore respects 
    agreement. 
%Every correct consensus protocol must respect Invariant~\ref{inv:decision}.
    \item \emph{Case} $r \neq r^\prime$: W.l.o.g. assume $r < r^\prime$. There 
    must be a round $\bar{r}$ such that $r \leq \bar{r} < r^\prime$ in which 
    all correct processes in $S_{\bar{r}}$ compute either $(D, v)$ or $(M, v)$ 
    outcomes during Phase 1 and after which $\exists q \in S_{\bar{r} + 1}$ 
    such that either $o_q = (D, v^\prime)$ or $o_q = (M, v^\prime)$ and 
    $v^\prime \neq v$. This is impossible: First, Phases 2 and 3 of round 
    $\bar{r}$ only allow processes to change their values to those of the 
    outcomes computed during Phase 1 of $\bar{r}$, which in our scenario will 
    be $v$. In other words, all values received by processes in $S_{\bar{r} + 
    1}$ in Phase 3 of round $\bar{r}$ will be $v$. Second, by the validity 
    property of the protocols in $\mcp$ no outcome $(D, v^\prime)$ or $(M, 
    v^\prime)$ can be computed after Phase 1 of round $\bar{r} + 1$ if all 
    processes start with proposal $v$.
\end{itemize}
By contradiction it must hold $v = v^\prime$.
\end{proof}

\subsection{Termination}
Termination encapsulates the liveness or progress requirement on consensus 
protocols by ensuring that eventually a decision is or can be made. In our 
model (Section~\ref{sec:mptc_model}), we stated two versions of termination, 
the definite and the probabilistic one. Notice that by definition, definite 
termination implies probabilistic termination.

Recall that in MPTC each consensus round is run by a subset of the system's 
processes. It is possible that in an infinite execution some correct processes 
may only execute a finite number of consensus protocols rounds and as a result 
not be able to decide. Note that even if at least one correct process decides, 
all correct processes can eventually learn this decision by having the decided 
processes broadcast a special decision message to all processes in $\mcn$, 
which in turn decide upon reception of that message. Therefore to guarantee 
definite (probabilistic) termination we simply need to show that eventually at 
least one correct process in $\mcn$ decides (with probability 1).

MPTC termination depends on both the guarantees provided by protocols in $\mcp$ 
and the properties of the interleavings of rounds of different protocols during 
MPTC execution. The guarantees of each $P \in \mcp$ allow us to reason about 
the properties that each MPTC round satisfies. We know that every round of $P$ 
must produce an outcome in $\{D, M, U\} \times \mcv$ provided that $\mca_P$ 
hold. If not, $P$ would violate termination. Note that $P$ should also satisfy 
termination under $\mca_\mcp$ as well, since $\mca_\mcp$ makes the same or 
stronger assumptions. If $P$ guarantees definite termination, then in any 
infinite execution of $P$ there must be at least one correct process $p \in 
\mcn$ that produces outcome $(D, v)$, $v \in \mcv$ in infinitely many rounds. 
If $P$ guarantees probabilistic termination then $\exists \epsilon \in (0, 1]$ 
such that in any infinite execution of $P$, infinitely many rounds have 
probability at least $\epsilon$ for at least one correct process to produce 
outcome $(D, v)$, $v \in \mcv$.

Given the previous guarantees provided by the protocols in $\mcp$ we can show 
the following:

\begin{lemma}
\label{lem:non_block}
Every correct process executing a round of MPTC under $\mca_\mcp$ eventually completes that round.
\end{lemma}

\begin{proof}
To show the above we need to ensure that each active participant of some round 
$r$ completes the all three Phases. Phase 1 completes by the termination 
property of each protocol in $\mcp$ with each correct process in $S_r$ 
computing an outcome with respect to the protocol specified by configuration 
$C_r$. Phases 2 and 3 rely on each correct process eventually receiving 
$\mcp_f - f$ messages which will occur due to our assumptions on network 
reliability and maximum number of failures and processes under attack.
\end{proof}

Lemma~\ref{lem:non_block} ensures no process participating in any of the MPTC 
Phases ever blocks. This \emph{non-blocking} property of MPTC rounds, however, 
is not sufficient to ensure progress. To reason about progress, we need to 
reason about the effect of interleaving rounds of different protocols on 
infinite executions. 

To reason about interleavings of consensus protocol rounds we need to reason 
about configuration sequences and thus about the properties of the 
unpredictable functions $F_S$ selected by dealer $T$ for each participant set 
$S$ during initialization. Consider the case where $\mcp$ contains two 
artificial protocols $P$, $P^\prime$ such that any process running $P$ can only 
decide during an odd round, while any process running $P^\prime$ can only 
decide on an even round. Assume a pathological infinite execution in which $C_r 
= (P, S_r)$ if $r$ is even and $C_r = (P^\prime, S_r)$ if $r$ is odd. If we can 
define $F_{S_r}$ for each round $r$ in such a way that the previous 
configuration sequence is generated during Phase 2 of MPTC, termination cannot 
be achieved.

%More formally, we call a protocol round \emph{possibly deciding} if there exists a run of that round under $\mca_\mcp$ in which at least one correct process can compute an outcome $(D, v)$ for some $v \in \mcv$. We call any infinite sequence of configurations from $\mcc$ an \emph{interleaving} in $\mcc$ and we define a conforming interleaving as one that contains infinitely many possible deciding rounds.

More formally, let the set of possible configurations, $\mcc$, be based on a 
valid $\mcp$. We define a finite sequence $\mci$ of configurations in $\mcc$ as 
\emph{conforming} if executing MPTC under $\mca_\mcp$ such that $\mci$ appears 
infinitely often, at least one correct process computes a decision outcome in 
infinitely many rounds. Any finite sequence of configurations in $\mcc$ that 
does not have the previous property is called \emph{non-conforming}. The case 
described in the previous paragraph is an example of such a non-conforming 
sequence. From now on we will refer to finite sequences of configurations in 
$\mcc$ as interleavings. 

%With this framework in mind it is trivial to see that:
%
%\begin{corollary}
%\label{cor:termination_impossibility}
%Given configurations set $\mcc$ based on some valid $\mcp$, if every interleaving in $\mcc$ is non-conforming, termination is impossible.
%\end{corollary}
%
%The previous statement creates another restriction on any implementation of MPTC and more specifically on the choice of $\mcp$ and definition of $\mcc$. In every implementation there must be at least one conforming interleaving in $\mcc$.

The previous notion of conforming interleavings raises another restriction on 
any implementation of MPTC and more specifically on the choice of $\mcp$. In 
every implementation there must be at least one conforming interleaving 
containing configurations of $\mcc$. We call an infinite execution of MPTC in 
which the corresponding configuration sequence contains infinitely many 
conforming interleavings, a \emph{conforming execution}.

In the random oracle model, which we assume in this work, the configurations 
sequences generated by the shared functions $F_S$ have the following property:

\begin{lemma}
\label{lem:cinterleavings}
Let $\mcc$ be a set of configurations based on a valid set of protocols $\mcp$ 
and let $F_S$ be an unpredictable function for each $S \in PS$ generated by 
$T$. Assuming a conforming interleaving $\mci$ exists in $\mcc$, any infinite 
sequence of configurations corresponding to an infinite execution of MPTC 
contains infinitely many occurrences of $\mci$.
\end{lemma}

\begin{proof}
Each $F_S$ used to generate the next configuration in Phase 2 of each MPTC 
round is based on the threshold coin-tossing mechanism described in 
Section~\ref{sec:mptc_model}. The unpredictability of this mechanism is based 
on the use of cryptographic hash functions. In the random oracle model, given 
some input $r$ these functions produce a value chosen uniformly at random from 
their co-domain. In other words, at each round $r$ there is a positive 
probability for each configuration $C \in \mcc$ to be selected as the next 
configuration to run. Therefore in an infinite execution, any interleaving in 
$\mcc$ appears infinitely often. Thus conforming interleaving $\mci$ appears 
infinitely often.
\end{proof}

By the previous Lemma we have that:

\begin{corollary}
\label{cor:cexecutions}
Assuming a conforming interleaving $\mci$ exists in $\mcc$, any infinite 
execution of MPTC is conforming.
\end{corollary}

We can now reason about the termination guarantees of MPTC, under the 
assumptions that $\mcp$ is valid, that there is a conforming interleaving in 
$\mcc$ and that $\mca_\mcp$ hold during MPTC's execution. Depending on the 
protocols executed under the configurations of a conforming interleaving, 
definite or probabilistic termination can be guaranteed. This is shown in the 
following lemmas:

\begin{lemma}
\label{lem:mptc_def_termination}
Let $\mci$ be a conforming interleaving in $\mcc$ that contains at least one 
round of a protocol satisfying definite termination. In any infinite execution 
of MPTC there is at least one correct process that makes a decision.
\end{lemma}

\begin{proof}
By Lemma~\ref{lem:cinterleavings} we have that $\mci$ will be executed 
infinitely often and so will the round of some protocol $P$ satisfying definite 
termination. By the termination guarantees of $P$ it must be the case that at 
least one correct process computes a decision outcome $(D, v)$ for some $v \in 
\mcv$ in infinitely many rounds. Therefore eventually at least one correct 
process running MPTC decides.
\end{proof}

We can state a similar lemma for probabilistic termination:

\begin{lemma}
\label{lem:mptc_prob_termination}
Let $\mci$ be a conforming interleaving in $\mcc$ that contains rounds of 
protocols satisfying probabilistic termination. In any infinite execution of 
MPTC there is at least one correct process that makes a decision with 
probability 1.
\end{lemma}

\begin{proof}
The proof is similar to that of Lemma~\ref{lem:mptc_def_termination} except 
from the fact that the rounds in which a decision can be made can only ensure 
that at least one process decides with probability 1.
\end{proof}

Finally, we need to show the security properties satisfied by MPTC under the 
previous assumptions. More specifically we need to show the following two 
properties:

\begin{lemma}
\label{lem:mptc_robustness}
\emph{Robustness}: It is computationally infeasible for $A$ to produce $r$ and 
$f+1$ shares of $r$ for any participant set $S$ such that the output of 
\textit{combine} given the previous shares and value as input is $F_S(r)$.
\end{lemma}

\begin{proof}
It follows directly from the robustness property of the threshold coin-tossing 
scheme we use for computing the next round's configuration in each MPTC round. 
The property is proven in~\cite{Cachin05}.
\end{proof}

\begin{lemma}
\label{lem:mptc_unpredictability}
\emph{Unpredictability}: Let $C_r^A$ be $A$'s prediction for $F_{S_r}(r)$ after 
learning at most $f$ function shares for $F_{S_r}(r)$ as well as any number of 
functions shares for $F_{S_{r^\prime}}(r^\prime)$ for arbitrary many $r^\prime 
< r$. Then the probability of $C_r^A = F_{S_r}(r)$ is at most $\frac{1}{|\mcc|} 
+ \epsilon$ where $\epsilon \in \mathbb{R}$ is negligible.
\end{lemma}

\begin{proof}
This property follows from the implementation of each $F_S$ as a $(\mcp_f, f+1, 
f)$ threshold coin-tossing scheme and from the extended unpredictability 
property of a sequence of coins produced by this scheme shown 
in~\cite{Cachin05}. By selecting the length of the sequence to be $m = \lceil 
\log |\mcc| \rceil$ we ensure that the probability of $A$ predicting the next 
configuration having compromised at most $f$ processes in the active 
participant set is $\frac{1}{2^m} + \epsilon \leq \frac{1}{|\mcc|} + \epsilon$ 
for negligible security parameter $\epsilon$.
\end{proof}

The previous termination discussion relies on the random oracle assumption we 
made earlier. While this assumption is important for supporting the 
unpredictability and robustness properties of our configuration generation 
scheme it is not necessary if such properties are not needed. If we wanted to 
drop this assumption, we would need to place additional restrictions on $\mcp$ 
and $\mcc$ to satisfy termination for MPTC. More specifically, we need to 
ensure that any interleaving in $\mcc$ is conforming. Under that assumption the 
termination results still hold since any infinite execution consists of 
conforming interleavings in which at least some correct process makes a 
decision.

%% file: byzantine.tex
\section{Extension to Byzantine Failures}
\label{sec:mptc_byzantine}

\subsection{Byzantine agreement model}
Our MPTC protocol can be extended to support byzantine agreement. Under 
this weaker adversary assumption, we call a process honest if it faithfully 
executes the protocol. An honest process may crash during the execution but up 
to the point of crash its execution does not deviate from the protocol 
description. Observe that in our previous model (Section~\ref{sec:mptc_model})
all processes where honest. We call a process correct if it is honest and 
eventually makes progress. This implies that a correct process never crashes. 
Faulty processes, on the other hand, may deviate arbitrarily from the protocol 
but cannot alter the secret state.

The specification of the protocols in $\mcp$ also change. Each protocol in 
$\mcp$ is now a \emph{Byzantine Fault-Tolerant} (BFT) agreement protocol. In 
this problem, the objective is for all honest processes to agree on the same 
value. The correctness criteria are as follows:

\begin{itemize}
    \item \emph{Agreement}: If two honest processes decide, they decide the 
    same value.
    \item \emph{Validity}: If an honest process decides value $v$, then $v$ was 
    proposed by at least some process.
    \item \emph{Unanimity}: If all honest processes propose the same value $v$, 
    then an honest process that decides must decide $v$.
    \item \emph{Termination (Definite)}: All correct processes must eventually 
    decide.
    \item \emph{Termination (Probabilistic)}: All correct processes must 
    eventually decide with probability~1.
\end{itemize}

Note that it now holds $\mcp_f \geq 3f + 1$. Also note that the byzantine 
failures assumption is now part of our model assumptions set $\mca$ and thus is 
included in $\mca_\mcp$ for any set of BFT protocols $\mcp$. The round outcomes 
framework described in Section~\ref{sec:mptc_model} still hold under the 
following modifications on invariants~\ref{inv:decision} 
and~\ref{inv:undecided}:

\begin{invariant}
\label{inv:byz_decision}
If there exists honest $p \in \mcn,\ r \in \mathbb{N}$ such that $o_p^r = (D, 
v)$ where $v \in \mcv$ then for each correct $q \neq p \in \mcn$ it holds that 
$o_q^r = (M, v)$ or $o_q^r = (D, v)$.
\end{invariant}

\begin{invariant}
\label{inv:byz_undecided}
If there exists honest $p \in \mcn$, $r \in \mathbb{N}$ such that $o_q^r = 
(U, v)$ for some $v \in \mcv$ then there is no honest $q \in \mcn$, $u \in 
\mcv$ such that $o_q^r = (D, u)$ and $u \neq v$.
\end{invariant}

This new version of invariants refer to honest processes since faulty processes 
may update their state arbitrarily at any point in time. Thus the processes' 
outcomes are meaningful only for honest and correct processes.

\subsection{Trusted dealer protocol}
Let $\mcp$ be a set of BFT protocols that tolerate up to $f$ failures. The 
trusted dealer initialization protocol now becomes:

\begin{enumerate}
	\item The dealer assigns an identity for each $p \in \mcn$ using a public -
	key cryptography scheme. It generates a public-private key pair, 
	$(\textit{public}_p, \textit{private}_p)$ for each $p \in \mcn$.
    \item For each $S \in PS$ the dealer picks a function $F_S: \mathbb{N} 
    \rightarrow \mcc$ and generates a verification key $VK_S$ and for each $q 
    \in S$ a secret share, $h_S^q$ and a verification key $VK_S^q$.
    \item $T$ distributes shares and keys to processes over secure channels. 
    $\forall S \in PS$ each process $p \in S$ receives $h_S^p$, $VK_S$, and 
    $VK_S^q$, $\forall q \in S$. In addition, each process $p$ receives 
    $(\textit{public}_p, \textit{private}_p)$ and the public keys of all other 
    processes.
\end{enumerate}

In the byzantine case, we need to use a verifiable secret sharing scheme 
like~\cite{Feldman87} since we need a way to ensure that invalid 
function shares created by faulty processes can be identified and discarded by 
honest ones. In such schemes, a verification function is specified which 
typically works by receiving a value $r$, a share of $F_S(r)$ and some 
verification keys that depend on $F_S$ and the secret share used to generate 
the previous function share and outputs 1 or 0 indicating whether the share 
provided is a valid share of $F_S(r)$ or not. We define our verification 
function, \emph{verify}, as follows:

\begin{equation}
\textit{verify}: \mathbb{N} \times \mcf \times \mck^2 \rightarrow \{0, 1\} 
\nonumber
\end{equation}

\noindent where $\mck$ is the space of verification keys. Note that the 
\emph{verify} above is modeled after the share verification algorithm presented 
in~\cite{Cachin05}. Given a function share $F^p_S(r)$, it receives two 
verification keys, $VK_S$, and $VK_S^p \in \mck$, the first produced using 
$F_S$ and the second using $h^p_S$. 

$T$ can generate these verification keys using a secure cryptographic hash 
function, that is a function that is easy to compute but computationally 
infeasible to reverse. Feldman~\cite{Feldman87} provided some example functions 
with this property, with RSA being one of them. We abstract away such details 
and denote by $\textit{hash}: \{0, 1\}^* \rightarrow \mck$ a function that has 
this property. $T$ can use \emph{hash} to generate the above 
mentioned verification keys during phase 1 (after using split) as follows:

\begin{eqnarray}
VK_S &=& \textit{hash}(F_S), \forall S \in PS \nonumber \\
VK_S^p &=& \textit{hash}(h_S^p) \forall p \in S, \forall S \in PS \nonumber
\end{eqnarray}

Note that our \emph{hash} function takes different types of input, such 
as $[\{0, 1\}^* \rightarrow \mcc]$ and $\mcs$, but both are bit strings and so 
are the elements of $\mck$ and $\mcf$. While we are using different domain 
notations to distinguish between functions, secret and function shares, and 
verification keys, any implementation of these schemes is working with bit 
strings.

\subsection{Byzantine Tolerant MPTC}

The Byzantine version of our MPTC protocol is similar to that in 
Section~\ref{ssec:prt_dsc}. The processes maintain the same state as in 
Section~\ref{ssec:prt_dsc} plus the cryptographic keys generated by the dealer. 
All messages are signed using these keys. Messages are of the form $\langle 
\textit{Phase number}$, $\textit{process id}$, $\textit{round}$, 
\textit{signature}, $\ldots \rangle$ and processes ignore messages with invalid 
signatures or messages not destined to them. The protocol operates in rounds, 
the configuration of the first round, $C_0$, is determined by $T$ and is known 
by all processes and each round proceeds as follows:

\begin{itemize}
    \item \textbf{Phase 1}: Each process $p \in S_r$ runs a round of the 
    BFT protocol specified by $C_r$. 
    %Once $p$ computes an outcome $o_p$ for round $r$, it goes to Phase 2.
    Let $o_p$ be $p$'s outcome for round $r$. If $o_p = (D, v)$, then process 
    $p$ updates $\textit{proposal}_p=v$, decides $v$ and never updates $o_p$ 
    and $\textit{proposal}_p$ again in any future round. If $o_p = (M, v)$, 
    then $p$ updates $\textit{proposal}_p=v$. Regardless of $o_p$'s value, $p$ 
    goes to Phase 2.
    \item \textbf{Phase 2}: 
    \begin{itemize}
        \item \textit{Step 1}: Each $p \in S_r$ computes function share 
        $F_{S_r}^p(r) = GFS(h_{S_r}^p, r)$ and sends a Phase 2 message 
        $m = \langle 2, p, r_p, o_r^p$, $\textit{sign}_m, F_S^p(r) \rangle$ to 
        all processes in $S_r$, where $\textit{sign}_m$ is the signature of 
        message $m$. Then $p$ waits for $\mcp_f - f$ Phase 2 messages with 
        valid function shares from processes in $S_r$. Each process receiving a 
        Phase 2 message checks the validity of the function share contained 
        within using \textit{verify}$(r, F_S^p(r), VK_S, VK_S^p)$. Once $p$ 
        receives enough such messages from some $Q \subseteq S_r$, it proceeds 
        to Step 2.
        \item \textit{Step 2}: If $o_p = (U, *)$ where $*$ can be any value in 
        $\mcv$, then $p$ updates its proposal to a value $v$, selected 
        arbitrarily from the outcomes contained in the received Phase 2 
        messages. It also updates $o_p = (U, v)$.
%        \item \textit{Step 2}: Once it receives Phase 2 messages from a set of 
%        $\mcp_f - f$ processes, $Q \subseteq S_r$ and if $o_p = (U, *)$ it 
%        updates its outcome and proposal as described below. Let $R$ denote
%        the set of outcomes received:
%        \begin{itemize}
%	        \item Case 1: If there are at least $f+1$ outcomes in $R$ of the 
%	        form $(D, v)$ for some value $v \in \mcv$, then process $p$ updates 
%	        $\textit{proposal}_p=v$, decides $v$, sets its outcome $o_p = (D, 
%	        v)$ and never updates $o_p$ and $\textit{proposal}_p$ again in any 
%	        future round.
%	        \item Case 2: If $\forall o \in R$ it holds $o = (M, v)$ for some 
%	        $v \in \mcv$ then $o_p = (M, v)$ and $\textit{proposal}_p = v$.
%	        \item Case 3: If Case 1 did not apply and $(U, *) \in R$ where $*$ 
%	        can be any value in $\mcv$ or $\exists o, o^\prime \in R$ and $v, 
%	        v^\prime \in \mcv$ such that $(o = (M, v) \wedge o^\prime = (M, 
%	        v^\prime)) \wedge (v \neq v^\prime)$, then $p$ selects an 
%	        outcome $(*, v) \in R$ where $*$ can be any value in $\{M, U\}$ and 
%	        $v$ is the most frequent value in $R$, breaking ties arbitrarily. 
%	        $p$ then updates $o_p = (U, v)$ and $\textit{proposal}_p = v$.
%        \end{itemize}
        \item \textit{Step 3}: Let $F_S^Q(r)$ be the set of valid function 
        shares received by processes in $Q$. $p$ computes the configuration of 
        the next round, $r+1$, as $C_{r+1} = \textit{combine} (F_{S_r}^Q(r), 
        r)$ and moves on to Phase 3.
    \end{itemize}
    \item \textbf{Phase 3}: Each $p \in S_r$ sends a Phase 3 message $m = 
    \langle 3, p, r_p, o_p, \textit{sign}_m, C_{r+1} \rangle$ to each process 
    in $S_{r+1}$. $p$ update its state: $r_p = r+1$, $c_p = C_{r+1}$ and if it 
    is still undecided, it updates its outcome $o_p = \bot$. Each process $q 
    \in S_{r+1}$ that receives Phase 3 messages with the same configuration 
    value, $C_{r+1}$, from $\mcp_f - f$ processes, updates its proposal as 
    described below. Let $R$ denote the set of outcomes received:
    \begin{itemize}
        \item Case 1: If there are at least $f+1$ outcomes in $R$ of the 
        form $(D, v)$ for some value $v \in \mcv$, then process $q$ updates 
        $\textit{proposal}_q=v$, decides $v$, sets its outcome $o_q = (D, 
        v)$ and never updates $o_q$ and $\textit{proposal}_q$ again in any 
        future round.
        \item Case 2: If $\exists v \in \mcv$ such that at least $f+1$ outcomes 
        in $R$ are $(D, v)$ or $(M,v)$ then $\textit{proposal}_q = v$.
        \item Case 3: Otherwise, $q$ selects an outcome $(*, v) \in R$ where 
        $*$ can be any value in $\{M, U\}$ and $v$ is the most frequent value 
        in $R$, breaking ties arbitrarily. $q$ then updates 
        $\textit{proposal}_p = v$.
    \end{itemize}
    Then $q$ sets $r_q = r+1$, $c_q = C_{r+1}$ and if it is still undecided, it 
    sets $o_q = \bot$. Finally, it starts the next round.
\end{itemize}

The previous protocol runs for an unbounded number of rounds, like our main 
MPTC protocol, and eventually reaches a state in which at least $f+1$ correct 
processes decide and thus all correct processes can learn that decision.

\subsection{Correctness discussion}

The byzantine version of MPTC must satisfy the properties described above as 
well as robustness and unpredictability. We again assume a valid $\mcp$ 
containing correct BFT protocols, MPTC execution under $\mca_\mcp$ as well as a 
set of configurations $\mcc$ such that a conforming interleaving exists. The 
arguments for all of these properties are similar to those in 
Section~\ref{ssec:mptc_proofs}. In fact, the termination arguments for the 
byzantine MPTC are exactly the same. Robustness and unpredictability are also 
exactly the same and our arguments are mainly borrowed from the corresponding 
arguments of the Byzantine agreement protocol of~\cite{Cachin05}. For the rest 
of the properties we have the following results:

\begin{lemma}
\label{lem:byz_mptc_validity}
Byzantine MPTC satisfies validity.
\end{lemma}

\begin{proof}
By correctness of the underlying BFT, any decision made by an honest process 
during Phase 1 respects validity. A decision during Phase 3 can only 
occur if at least $f+1$ processes send the same decision outcome $(D, v)$ for 
some $v \in \mcv$. For that to happen at least one honest process must have 
computed that outcome during Phase 1. Thus that value must have been proposed 
by some process.
%It holds trivially due to the fact that the first decision by an honest process is done during Phase 2 and by validity of the BFT run during that round it must have been proposed.
\end{proof}

\begin{lemma}
\label{lem:byz_mptc_agreement}
Byzantine MPTC satisfies agreement.
\end{lemma}

\begin{proof}
Note that Phase 1 respects agreement by the correctness of the protocols in 
$\mcp$. Observe that honest processes can only decide during Phase 3 if 
at least $f+1$ processes have decided the same value. This means that if an 
honest process decides $v$ during Phase 3 of MPTC at least one honest 
process has already decided $v$. Using similar arguments to those in 
Lemma~\ref{correctness:agreement}, we can show that at any point during MPTC 
execution any two honest processes that decide must decide the same value.
\end{proof}

\begin{lemma}
\label{lem:byz_mptc_unanimity}
Byzantine MPTC satisfies unanimity.
\end{lemma}

\begin{proof}
Assume that all honest processes are initialized with the same value $v \in 
\mcv$. Since the BFT protocols in $\mcp$ are correct, honest processes getting 
into Phase 1 with the same value, $v$, can only compute outcomes $(D, v)$, $(M, 
v)$ and $(U, v)$. Otherwise it would mean that honest processes can change 
their proposals between rounds of the BFT protocol in question (e.g., via the 
influence of the faulty processes) and eventually decide a different value, 
which would violate unanimity of the BFT protocol. Therefore, at the end of 
Phase 1 all honest processes will end up with outcomes containing $v$. $v$ will 
be the most frequent value in any subset of $\mcp_f - f$ outcomes since $\mcp_f 
\geq 3f + 1$ and at least $2f+1$ processes are correct. As a result Phases 2 
and 3 will have all honest processes updating their outcomes and proposals to 
$v$. Consequently, all honest processes will eventually move to subsequent 
rounds with $v$ as their proposal and thus $v$ will be the only value that can 
be decided by some honest process.
\end{proof}

%% file: implementation_transitions.tex
\section{Implementation Details}
\label{sec:mptc_implementation_details}

\subsection{MPTC protocol implementation}
\label{sec:protocol_implementation}

To implement MPTC, we need to decide on the following parameters:

\begin{itemize}
    \item The choice of protocol set $\mcp$.
    \item The set of possible configurations $\mcc$, which specifies not only 
    the protocol of each round but also its initialization.
    \item The configuration selection of functions $F_S$, $\forall S \in PS$, 
    generated by the trusted dealer.
    \item The implementation of the \textit{split} function.
    \item The implementation of \textit{GFS} and \textit{combine} used during 
    the main protocol.
\end{itemize}

Our set of protocols, $\mcp$, contains only a single consensus protocol, an 
optimized version of single decree Paxos~\cite{Lamport98}. A round of single 
decree Paxos operates in 2 phases. In the first phase an active leader/proposer 
gets elected and in the second the active leader makes its proposal to the rest 
of the processes who may accept the proposal. If the proposal gets accepted by 
a majority of acceptors, it gets decided. Paxos tolerates $f$ crash failures 
using $2f+1$ processes and thus $\mcp_f = 2f+1$.

Similar to the implementation of Turtle Consensus~\cite{NvR15}, we avoid 
electing leaders in each round by using a parameterized version of single 
decree Paxos in which each round comes with a predetermined leader known to all 
active participants. Assuming an ordering of the processes executing the 
protocol, we can set the leader to be the process whose position in that order 
is equal to round number modulo $2f+1$. Under normal execution conditions, the 
previous optimization yields the following benefits: 1) a decision within a 
single round-trip of communication since it directly executes Phase 2 of Paxos, 
and 2) performance unaffected by contention since the active leader is the only 
proposer.

For handling failures, our optimization assumes the weakest failure detector, 
$\diamond \mcw$, presented in~\cite{CHT96}. If a leader failure is suspected by 
$\diamond \mcw$, then the suspecting processes will complete that round using 
$M$ or $U$ outcomes depending on whether they have received a proposal or not, 
respectively. We implement $\diamond \mcw$ like we did in Turtle Consensus, 
using timeouts with exponentially increasing timeout periods when processes are 
inaccurately suspected. This way we ensure that there will be enough rounds 
executed ``concurrently'' by sufficiently many processes which is critical for 
ensuring termination in our Paxos variant.

In the description above we did not specify how the states of processes running 
Paxos are turned into outcomes at the end of a Paxos round. The process 
executing as the active leader can either decide the value it proposes during 
Phase 2, say $v$, or fail to do so due to either failing or suspecting a 
majority of acceptors. In the first case, it computes outcome $(D,v)$. In the 
second case and if the leader did not fail, it will timeout knowing that if a 
decision was made it must have been for its own proposal, thus computing 
outcome $(M, v)$, where $v$ is its proposal in the beginning of the round. 
Similarly, a process running as an acceptor can end a round either having 
accepted the active leader's proposal, thus computing outcome $(M,v)$, or 
timing out without having received any proposal, in which case it does not know 
anything about the decision progress. In this latter case, we need to ensure 
that if the acceptor becomes the next round's active leader, it will make a 
proposal consistent with already accepted proposals. The way original Paxos 
achieves this, is through its phase 1. As a result our variant requires 
processes exiting the Paxos round without knowledge of the round's decision 
state to retrieve this knowledge from the rest of the processes.

To avoid incurring another round of communication in our Paxos variant we 
piggyback this decision state retrieval onto Phase 2 of MPTC. Timed out 
processes can use the set of outcomes received to update their proposal. The 
update procedure is similar to the one used by processes that have computed 
outcome $(U, *)$. The main difference is that processes without knowledge about 
the outcome of the round send a special ``unknown'' outcome. These ``unknown'' 
outcomes are ignored by receiving processes unless all outcomes received during 
Phase 2 of MPTC are ``unknown'', in which case no decision could have been 
made. In the latter case, the receiving process updates its outcome to $(U, 
v^\prime)$, where $v^\prime$ is the receiver's proposal at the beginning of the 
round, and acts as in Phase 2, Step 2, Case 3. The above optimization ensures 
that if the predetermined leader decides value $v$ at some round $r$ and a 
failure prevents that decision to be learnt in $r$, then all processes during 
Phase 2 of $r$ will receive at least one $(M,v)$ outcome. The receiving 
processes will be forced to adopt $v$ and thus future proposals can only be 
about $v$.

%Unfortunately, the previous variant is not sufficient to ensure correctness. 
%The reason is that phase 1 of original Paxos ensures that future leaders 
%elected when failures or large delays occur, will make proposals that are 
%consistent with previously accepted values. To provide similar guarantees we 
%piggyback, phase 1 of Paxos during Phase 2 of MPTC. Instead of having processes 
%choose an outcome arbitrarily in Phase 2, Step 2, Case 3 we use the following 
%modification: If all $(M, *)$ received outcomes have the same value, the 
%process picks that value. Otherwise it picks the value with the highest 
%popularity breaking ties arbitrarily. The previous modification ensures that if 
%the predetermined leader decides value $v$ at some round $r$ and a failure 
%prevents that decision to be learnt in $r$, then all processes during Phase 2 
%of $r$ will receive at least one $(M,v)$ outcome. The receiving processes will 
%be forced to adopt $v$ and thus future proposals can only be about $v$.

Our set of configurations is simply $\mcc = \{(S, P) | S \in PS\}$ where $P \in 
\mcp$ is the prior Paxos variant. We assume that for every 
participant set there is an ordering of the processes in it. This is easy to 
achieve using the processes unique identifiers, and it facilitates the 
selection of leader of each configuration without having to define additional 
initialization information in each configuration. Note that since all rounds 
execute the same correct consensus protocol, all possible interleavings in 
$\mcc$ are conforming.
%Depending on the current round $r$, the leader in $C_r$ will be the process in $S_r$ ordered in position $(r \mod 2f+1)$.

Observe that in contrast to our prior work on Turtle Consensus we use the same 
protocol across configurations. In Turtle Consensus~\cite{NvR15}, 
different configurations used the same $2f+1$ set of processes. As a result, 
the adversary could try to track the current leader within that set of 
processes even if the leader changed across different configurations. 
Therefore, a competent adversary could eventually locate and force Turtle 
Consensus rounds to fail which can lead to very poor performance. For that 
reason, we kept switching between a leader-based (Paxos) and fully 
decentralized (Ben-Or) consensus protocols across configurations to prevent the 
adversary of exploiting the leader vulnerability. A side-effect of that 
approach, however, was that by falling back to a less efficient protocol (Ben-
Or) we only achieved sub-par performance compared to the graceful execution 
using only Paxos rounds. With MPTC we do not need to employ such tactics since 
the adversary now needs to scan through $|\mcn| \gg f$ processes before it can 
identify the leader of our Paxos configuration.
%and while it is still possible, our evaluation suggests that the cost of changing configuration does not substantially degrade the overall performance.

The remaining parameters of our implementation are related to the cryptographic 
framework assumed by our protocol. While we did not implement this framework 
for our evaluation, we describe for completeness how we could do so in the 
following paragraphs. For the actual implementation that we evaluate in 
Section~\ref{sec:mptc_evaluation}, we assume that all participant sets use the 
same unpredictable function given to all processes via a configuration file. 
This file simply defines a sequence of configurations, one for each round, that 
processes move to in a round-robin fashion, that is use the first configuration 
in round 0, the second in round 1, etc. Once the last configuration in the 
sequence is used, the processes loop back to the first one and continue from 
there. We emulate the restrictions that the cryptographic framework imposes on 
the adversary by assuming that only the processes involved in rounds $r$ and $r+1$ can learn $C_{r+1}$ and only after Phase 2 of round $r$ completes.

A potential implementation for the unpredictable functions is having $F_S = F$, 
$\forall S \in PS$, where $F$ is derived from the threshold coin-tossing scheme 
implementation based on Diffie-Hellman and presented in~\cite{Cachin05}. This 
approach hashes the input value, $r$, using a cryptographically secure hash 
function, modeled as a random oracle, and raises the result to a secret 
exponent. This exponent is shared among the processes using Shamir's secret 
sharing~\cite{Shamir79}. Finally, the result is further hashed to obtain the 
value of $F(r)$ using another cryptographically secure hash function.

Function \textit{split} used by the dealer during initialization can be 
implemented using Shamir's secret sharing as mentioned above. Function 
\textit{GFS} can be implemented by having each process $p$ hashing the input 
round number $r$ and raising it to its secret share of the exponent it received 
from the dealer. Finally, \textit{combine} for (byzantine) MPTC simply 
multiplies a set of $f+1$ distinct (valid) shares and hashes the result. Note 
that the \textit{combine} computation is slightly more complicated and a 
detailed version of its implementation can be found in~\cite{Cachin05}. Also 
note that one can alternatively use any non-interactive threshold-signature 
scheme with the property that there is only one valid signature per message, 
like the RSA-based scheme of Shoup~\cite{Shoup00}. We can then obtain the value 
of the function by hashing the resulting signature computed by $f+1$ signature 
shares on a message containing input $r$.

\subsection{MPTC-based state machine replication}
\label{sec:mptc_smr}

In this section we describe in detail how we implemented SMRP using our MPTC implementation. Messages exchanged in this implementation are of the form:

\begin{equation}
\langle \textit{type}, \textit{src}, \textit{dst}, (\textit{content-
attribute-1, content-attribute-2}, \ldots) \rangle \nonumber
\end{equation}

\noindent where \textit{type} $\in \{\texttt{REQUEST}, \texttt{RESPONSE}, 
\texttt{DECISION}, \texttt{RECONFIGURATION}\}$, \textit{src}, \textit{dst} $\in 
\mathbb{N}$ are the ids of the communicating processes, and \textit{content-
attribute-x}, for $x \in \mathbb{N}$ constitute the message payload. The 
different types of messages are as follows: A \textit{request} message is sent 
by a client to a participant and may be forwarded between participants. It 
carries a deterministic operation to be executed by the service. A 
\textit{response} message is sent by a replica to a participant which then 
forwards it to clients and contains the result of the execution of an operation 
issued via a request by the client. A \textit{decision} message is sent by a 
participant to a replica, and carries a request along with its order of 
execution. Finally, \textit{reconfiguration} messages are sent between 
participants and are used to update the configuration of the MPTC execution.

\subsubsection{Clients}
Clients are uniquely identified processes whose identities are independent from 
those of the participants and the replicas. They connect to at least $f+1$ 
participants to which they send requests of the form: 

\begin{equation}
(\textit{client-id}, \textit{request-number}, \textit{command}) \nonumber
\end{equation}

\noindent where \textit{request-number} is a unique identifier for a particular 
request sent by this client and \textit{command} is an application-specific 
description of a deterministic operation and of any arguments that operation 
requires. Note that each pair $(\textit{client-id}, \textit{request-number})$ 
uniquely identifies a request received by the state machine and will be used by 
participants and replicas to track requests that are new, under processing, or 
executed. Clients therefore maintain the following state:

\begin{itemize}
	\item \textit{cid} $\in \mathbb{N}$, which is initialized with a unique 
	value identifying the client.
	\item \textit{rsn} $\in \mathbb{N}$, which is initialized to 0, incremented 
	each time the client issues a new request and uniquely identifies requests 
	for a given client.
	\item \textit{pending-requests} $\subseteq \mathbb{N}$, which is an 
	initially empty set that stores ids of requests the client has sent, but 
	has not yet received response.
\end{itemize}

Our clients are modeled as state machines with the following transitions:

\begin{enumerate}
	\item[T1:]\textbf{Precondition}: There is a command \textit{cmd} that needs 
	to be executed\\
	\textbf{Action}: Send message $\langle\textit{cid}, \textit{pid}, 
	\texttt{REQUEST}, (\textit{cid}, \textit{rsn}, \textit{cmd})\rangle$ to 
	each participant \textit{pid} that \textit{cid} is connected to; add 
	\textit{rsn} to 	\textit{pending-requests} and update $\textit{rsn} = 
	\textit{rsn} + 1$.
	\item[T2:]\textbf{Precondition}: Received message $\langle\textit{pid}, 
	\textit{cid}, \texttt{RESPONSE}, (\textit{rsn}^\prime, 
	\textit{result})\rangle$\\
	\textbf{Action}: If $\textit{rsn}^\prime \notin \textit{pending-requests}$ 
	ignore; otherwise remove $\textit{rsn}^\prime$ from 
	\textit{pending-requests}.
\end{enumerate}

\subsubsection{Participants}
The participants receive requests from clients and are responsible for ordering 
these requests and send them for execution to the replicas. To achieve this 
they spawn MPTC instances, one for each request that needs to be ordered. Each 
instance has its own identifier and decided requests are ordered according to 
the identifiers of the MPTC instances that decided them. Only one participant 
set can be active at any point in time. The participants of this set are 
responsible for creating consensus instances. The active participant set may 
concurrently run different instances of MPTC to deal with multiple requests. As 
in the case of Turtle Consensus~\cite{NvR15}, we consider a global threshold 
$W$ that limits the number of concurrent consensus instances that have not yet 
decided. Requests that arrive when $W$ concurrent MPTC instances are running 
are queued and processed when some of the running instances complete.

Note that in most common SMRPs clients typically send their requests to 
replicas directly, which are responsible for instantiating consensus instances 
to these requests. By disabling communication between clients and replicas we 
prevent clients from launching DoS attacks on the replicas. This approach does 
however have the issue that when a decision is made and the corresponding 
request is executed, two hops of communication are needed for the response to 
arrive at the clients. In addition, it burdens participants with implementing 
functionality beyond running consensus that is commonly offered by replicas, 
like keeping track of the values already decided, maintain state for ongoing 
instances and other aspects of the implementation, which makes it a less 
efficient approach. In this implementation, however, we are primarily 
interested in building an attack-tolerant SMRP and thus made the choice of 
communication pattern we described above.

Each participant is connected to all replicas and zero or more clients. Apart 
from the state required to run MPTC described in 
Section~\ref{sec:mptc_protocol}, each participant additionally maintains the 
following state:

\begin{itemize}
	\item \textit{pid} $\in \mcn$, initialized with the identifier of the 
	participant.
	\item \textit{next-instance} $\in \mathbb{N}$, initialized to 0, stores the 
	id of the next instance that will be created.
	\item \textit{configuration} $\in \mcr$, initialized to $C_0$ provided by 
	the dealer, stores the configuration that $pid$ considers active.
	\item \textit{requests} $\subseteq \mathbb{N}^2 \times \textit{Ops}$, which 
	is an initially empty set that stores requests that have been received but 
	not yet decided. \textit{Ops} denotes the space of commands and it is 
	application-specific.
	\item \textit{instances} $\subseteq \mathbb{N}^3$, which is an initially 
	empty set that stores for each running instance the id of the instance as 
	well as the request identifier ($\textit{cid}\ ^\prime$, \texttt{rsn}); the 
	request 	identifier is \textit{pid}'s input proposal for that MPTC instance.
	\item \textit{rstate} $\in \{\texttt{TRUE}, \texttt{FALSE}\}$, initialized 
	to \texttt{FALSE}, indicates whether the participant is currently under 
	reconfiguration.
	\item \textit{responses} $\subseteq \mathbb{N}^3$, which is an initially 
	empty set that stores mappings of requests and processes from which they 
	were received. It is primarily used to forward responses back to the 
	processes that sent or relayed requests.
\end{itemize}

Let \textit{configuration.participants} denote the participant set of 
\textit{configuration} and let \textit{configuration.round} denote the round in 
which \textit{configuration} is used. Aside from the transitions determined by 
the MPTC instances a participant may be running, it additionally performs the 
following transitions:

\begin{enumerate}
    \item[T1:]\textbf{Precondition}: Received message $\langle 
    \texttt{REQUEST}, \textit{pid}\ ^\prime, \textit{pid}, (\textit{cid}, 
    \textit{rsn}, \textit{cmd}) \rangle$ where $\textit{pid}\ ^\prime$ is 
    either a client or another participant, $\textit{rstate} = \texttt{FALSE}$, 
    $\textit{pid} \notin \textit{configuration.participants}$ and 
    $(\textit{pid}\ ^\prime, \textit{cid}, \textit{rsn}) \notin 
    \textit{responses}$\\
	\textbf{Action}: Add ($\textit{pid}\ ^\prime$, \textit{cid}, \textit{rsn}) 
	in \textit{responses} and send $\langle \texttt{REQUEST}, \textit{pid}, p, 
	(\textit{cid}, \textit{rsn}, \textit{cmd}) \rangle$ to each $p \in 
	\textit{configuration.participants}$.
	\item[T2:]\textbf{Precondition}: Received message $\langle 
	\texttt{REQUEST}, \textit{pid}\ ^\prime, \textit{pid}, (\textit{cid},$ 
	$\textit{rsn}, \textit{cmd}) \rangle$ where $\textit{pid}\ ^\prime$ is 
	either a client or another participant, $\textit{rstate} = \texttt{FALSE}$, 
	$\textit{pid} \in \textit{configuration.participants}$, $\nexists 
	(\textit{cid}, \textit{rsn}, *) \in \textit{requests}$ and $(\textit{pid}\ 
	^\prime, \textit{cid}, \textit{rsn}) \notin \textit{responses}$\\
	\textbf{Action}: Add (\textit{cid}, \textit{rsn}, \textit{cmd}) in 
	\textit{requests} and ($\textit{pid}\ ^\prime$, \textit{cid}, \textit{rsn}) 
	in \textit{responses}. Send $\langle \texttt{REQUEST}, \textit{pid}, p, 
	(\textit{cid}, \textit{rsn}, \textit{cmd}) \rangle$ to each $p \in 
	\textit{configuration.participants} \setminus \{\textit{pid}\}$. If 
	$W >$ $|\textit{instances}|$ then create an MPTC instance with instance id, 
	\textit{next-instance} and proposal (\textit{cid}, \textit{rsn}). Then add 
	(\textit{next-instance}, \textit{cid}, \textit{rsn}) to \textit{instances} 
	and update $\textit{next-instance}$ $= \textit{next-instance} + 1$. Finally, 
	run round \textit{configuration.round} of the new MPTC instance.
    \item[T3:]\textbf{Precondition}: Received message $\langle 
    \texttt{REQUEST}, \textit{pid}\ ^\prime, \textit{pid}, (\textit{cid}, 
    \textit{rsn}, \textit{cmd}) \rangle$ where $\textit{pid}\ ^\prime$ is 
    either a client or another participant, $\textit{rstate} = \texttt{FALSE}$, 
    $\textit{pid} \in \textit{configuration.participants}$, $\exists 
    (\textit{cid}, \textit{rsn}, *) \in \textit{requests}$ and $(\textit{pid}\ 
    ^\prime, \textit{cid}, \textit{rsn}) \notin \textit{responses}$\\
	\textbf{Action}: Add ($\textit{pid}\ ^\prime$, $\textit{cid}$, 
	\textit{rsn}) in \textit{responses}.
	\item[T4:]\textbf{Precondition}: Received message $\langle 
	\texttt{REQUEST}, \textit{pid}\ ^\prime, \textit{pid}, (\textit{cid}, 
	\textit{rsn}, \textit{cmd}) \rangle$ where $\textit{pid}\ ^\prime$ is 
	either a client or another participant, and $\textit{rstate} = 
	\texttt{TRUE}$\\
	\textbf{Action}: If $\nexists (\textit{cid}, \textit{rsn}, *) \in 
	\textit{requests}$ add (\textit{cid}, \textit{rsn}, \textit{cmd}) in 
	\textit{requests}. If $(\textit{pid}\ ^\prime, \textit{cid}, \textit{rsn}) 
	\notin \textit{responses}$ add ($\textit{pid}\ ^\prime$, \textit{cid}, 
	\textit{rsn}) in \textit{responses}.
	\item[T5:]\textbf{Precondition}: Received message $\langle 
	\texttt{RESPONSE}, \textit{rid}, \textit{pid}, (\textit{cid}, \textit{rsn}, 
	\textit{result}) \rangle$ from any replica or participant \textit{rid} and 
	$\exists (*, \textit{cid}, \textit{rsn}) \in \textit{responses}$\\
	\textbf{Action}: For each $(\textit{pid}\ ^\prime, \textit{cid}, 
	\textit{rsn}) \in \textit{responses}$, send $\langle \texttt{RESPONSE}, 
	\textit{pid}, \textit{pid}\ ^\prime, (\textit{cid}$, $\textit{rsn}$, 
	$\textit{result}) \rangle$ and remove ($\textit{pid}\ ^\prime$, 
	\textit{cid}, \textit{rsn}) from \textit{responses}.
	\item[T6:]\textbf{Precondition}: On round completion of some instance 
	$(\textit{instance}, \textit{cid}, \textit{rsn}) \in \textit{instances}$ 
	with outcome $(D, (\textit{cid}, \textit{rsn}))$ and (\textit{cid}, 
	\textit{rsn}, \textit{cmd}) $\in$ \textit{requests}\\
	\textbf{Action}: Send $\langle \texttt{DECISION}, \textit{pid}, 
	\textit{rid}, (\textit{instance}, \textit{cid}, \textit{rsn}, \textit{cmd}, 
	\textit{configuration}) \rangle$ to each $\textit{rid} \in \mcr$. Then 
	remove (\textit{instance}, \textit{cid}, \textit{rsn}) from 
	\textit{instances} and (\textit{cid}, \textit{rsn}, \textit{cmd}) from 
	\textit{requests}.
    \item[T7:]\textbf{Precondition}: $|\textit{instances}| = 0$, $\exists 
    (\textit{cid}, \textit{rsn}, \textit{cmd}) \in \textit{requests}$, 
    $\textit{rstate} = \texttt{FALSE}$ and $\textit{pid} \in 
    \textit{configuration.participants}$\\
	\textbf{Action}: Create an MPTC instance with instance id, \textit{next-
	instance} and proposal (\textit{cid}, \textit{rsn}). Then add 
	(\textit{next-instance}, \textit{cid}, \textit{rsn}) to \textit{instances} 
	and update $\textit{next-instance} = \textit{next-instance} + 1$. Finally, 
	run round \textit{configuration.round} of the new MPTC instance.
	\item[T8:]\textbf{Precondition}: On round completion of some instance 
	$(\textit{instance}, \textit{cid}, \textit{rsn}) \in \textit{instances}$ 
	with outcome $(M, (\textit{cid}, \textit{rsn}))$ or $(U, (\textit{cid}, 
	\textit{rsn}))$ and there are still instances that have not completed 
	\textit{configuration.round}\\
	\textbf{Action}: Update $\textit{rstate} = \texttt{TRUE}$.
	\item[T9:]\textbf{Precondition}: On round completion of some instance 
	$(\textit{instance}, \textit{cid}, \textit{rsn}) \in \textit{instances}$ 
	with outcome $(M, (\textit{cid}, \textit{rsn}))$ or $(U, (\textit{cid}, 
	\textit{rsn}))$, next configuration \textit{cfg} and there no more 
	instances that have not completed \textit{configuration.round}\\
	\textbf{Action}: Update $\textit{rstate} = \texttt{FALSE}$ and 
	$\textit{configuration} = \textit{cfg}$. Then send $\langle 
	\texttt{RECONFIGURATION}$, \textit{pid}, \textit{p}, (\textit{instances}, 
	\textit{requests}, \textit{configuration}, \textit{instance})$\rangle$ to 
	each $\textit{p} \in \textit{configuration}$. Finally, update 
	$\textit{instances} = \emptyset$ and $\textit{requests} =  \emptyset$.
	\item[T10:]\textbf{Precondition}: Received messages 
	$\langle$\texttt{RECONFIGURATION}, $p$, \textit{pid}, 
	($\textit{instances}_p$, $\textit{requests}_p$, 
	$\textit{configuration}^\prime$, $\textit{instance}_p) \rangle$ from a set 
	of participants $Q$, $|Q| = f+1$ such that \\
	$\textit{configuration}^\prime\textit{.round} > 
	\textit{configuration.round}$\\
	\textbf{Action}: Update $\textit{configuration} = 
	\textit{configuration}^\prime$, $\textit{next-instance} = \textit{max}
	(\textit{instance}_p \text{ for } p \in Q)$ and $\textit{requests} = 
	\textit{requests} \bigcup (\cup_{p \in Q} \textit{requests}_p)$. For each 
	$p \in Q$ and $(\textit{cid}, \textit{rsn}, \textit{cmd}) \in 
	\textit{requests}_p$ add ($p$, \textit{cid}, \textit{rsn}) in 
	\textit{responses}. For each $(\textit{instance}^\prime, \textit{cid}, 
	\textit{rsn}) \in \cup_{p \in Q} \textit{instances}_p$ use Phase 3 of MPTC 
	to update the proposal of $\textit{instance}^\prime$ and create an MPTC 
	instance $(\textit{instance}^\prime, \textit{cid}\ ^\prime, 
	\textit{rsn}^\prime)$ where $(\textit{cid}\ ^\prime, \textit{rsn}^\prime)$ 
	is the updated proposal. Finally, add $(\textit{instance}^\prime, 
	\textit{cid}\ ^\prime, \textit{rsn}^\prime)$ to \textit{instances}.
\end{enumerate}

Like in Turtle Consensus, MPTC is lazily instantiated for each slot. MPTC 
messages carry instance identifiers so incoming protocol messages are properly 
processed by the correct instance. If an instance has not yet been created, 
messages for that instance are queued and processed when it is created. 
Transition T2 ensures that if a request creates an instance to one of the 
active participants then unless the receiving participant fails, the remaining 
correct processes in the active participant set will also receive and create an 
instance for that request. Transition T7 ensures that if a proposed request 
does not get decided in one of the running instances, a new instance will be 
created for that proposal.

\subsubsection{Replicas}
The sets of replicas and participants are disjoint and $|\mcr| > f$ so at least 
one replica is always correct. Each replica in $\mcr$ is run by a different 
process and maintains a copy of state of the service implemented by the SMRP. 
Replicas are responsible for executing the commands issued through clients' 
requests in the order established by the participants. This order is determined 
by the instance id's associated with each decision received. As in~\cite{NvR15}, we model this ordering as a sequence of numbered slots 
with the first slot numbered as $0$. All replicas are initialized in the same 
state and since all commands are deterministic, executing commands in order at 
all replicas ensures that they all end up in the same state. To provide this 
functionality, replicas maintain the following state additional to the one 
related with the service SMRP implements:

\begin{itemize}
	\item \textit{rid} $\in \mcr$, initialized with the identifier of the 
	replica.
	\item \textit{execution-slot} $\in \mathbb{N}$, which is initialized 0 and 
	maintains the position in the ordering of commands that should be executed 
	next.
	\item \textit{decisions} $\subseteq \mathbb{N}^3 \times \textit{Ops} \times 
	\mcc$, which is an initially empty set of mappings between a slot and a 
	request as well as the configuration under which the request was decided.
	\item \textit{completed} $\subseteq \mathbb{N}^3$, which is an initially 
	empty set that stores mappings between slots and ids of requests that were 
	executed in those slots.
\end{itemize}

In addition, to those transitions defined by the deterministic state machine 
implemented by the SMRP, replicas perform the following transitions:

\begin{enumerate}
	\item[T1:]\textbf{Precondition}: Received message $\langle 
	\texttt{DECISION}$,  $\textit{pid}$, $\textit{rid}$, ($\textit{instance}$, 
	$\textit{cid}$, $\textit{rsn}$, $\textit{cmd}$, $C) \rangle$ where $C \in 
	\mcc$, $\nexists (\textit{instance}, *, *, *, *) \in \textit{decisions}$, 
	and $\textit{execution-slot} \leq \textit{instance}$\\
	\textbf{Action}: Add $(\textit{instance}$, $\textit{cid}$, $\textit{rsn}$, 
	$\textit{cmd}$, $C)$ in \textit{decisions}.
	\item[T2:]\textbf{Precondition}: $\exists (\textit{instance}, \textit{cid}, 
	\textit{rsn}, \textit{cmd}, C) \in \textit{decisions}$, and 
	$\textit{execution-slot} = \textit{instance}$, and $\nexists (*, 
	\textit{cid}, \textit{rsn}) \in \textit{completed}$\\
	\textbf{Action}: Execute $\textit{cmd}$ and let \textit{result} be the 
	outcome of the operation. Send $\langle \texttt{RESPONSE}, \textit{rid}$, 
	$\textit{pid}, (\textit{cid}, \textit{rsn}, \textit{result}) \rangle$ to 
	each $\textit{pid}$ in the participant set specified by $C$. Add 
	$(\textit{instance}, \textit{cid}$, $\textit{rsn})$ to \textit{completed} and 
	remove (\textit{instance}, \textit{cid}, \textit{rsn}, \texttt{cmd}, $C$) 
	from \textit{decisions}. Update $\textit{execution-slot}$ $=$ 
	$\textit{execution-slot} + 1$.
	\item[T3:]\textbf{Precondition}: $\exists (\textit{instance}, \textit{cid}, 
	\textit{rsn}, \textit{cmd}, C) \in \textit{decisions}$, and 
	$\textit{execution-slot} = \textit{instance}$, and $\exists (*, 
	\textit{cid}, \textit{rsn}) \in \textit{completed}$\\
	\textbf{Action}: Add (\textit{instance}, \textit{cid}, \textit{rsn}) to 
	\textit{completed} and remove (\textit{instance}, \textit{cid}, 
	\textit{rsn}, \textit{cmd}, $C$) from \textit{decisions}. Update 
	$\textit{execution-slot} = \textit{execution-slot} + 1$.
\end{enumerate}

In other words, replicas execute commands one slot at a time as soon as they 
become available. Decisions already received or executed are ignored or treated 
as no-op. Note that there is still the issue of garbage collecting executed 
requests in \textit{completed} that are no longer needed. In fact, 
\textit{completed} can get arbitrarily large the longer the system operates. 
One solution would be to only track for each client the largest request 
identifier used such that future decisions about an already executed request 
number for a given client can be ignored.

%% file: paper.bbl
\begin{thebibliography}{10}

\bibitem{AAKS14}
Dan Alistarh, James Aspnes, Valerie King, and Jared Saia.
\newblock Communication-efficient randomized consensus.
\newblock In Fabian Kuhn, editor, {\em Distributed Computing - 28th
  International Symposium, {DISC} 2014, Austin, TX, USA, October 12--15, 2014.
  Proceedings}, volume 8784 of {\em Lecture Notes in Computer Science}, pages
  61--75. Springer, 2014.

\bibitem{ANL01}
Tuomas Aura, Pekka Nikander, and Jussipekka Leiwo.
\newblock {DOS}-resistant authentication with client puzzles.
\newblock In {\em Security Protocols}, volume 2133 of {\em Lecture Notes in
  Computer Science}, pages 170--177. Springer Berlin Heidelberg, 2001.
\newblock URL: \url{http://dx.doi.org/10.1007/3-540-44810-1_22}, \href
  {http://dx.doi.org/10.1007/3-540-44810-1_22}
  {\path{doi:10.1007/3-540-44810-1_22}}.

\bibitem{BenOr83}
Michael Ben-Or.
\newblock Another advantage of free choice (extended abstract): Completely
  asynchronous agreement protocols.
\newblock In {\em Proceedings of the Second Annual ACM Symposium on Principles
  of Distributed Computing}, PODC '83, pages 27--30, New York, NY, USA, 1983.
  ACM.
\newblock URL: \url{http://doi.acm.org/10.1145/800221.806707}, \href
  {http://dx.doi.org/10.1145/800221.806707} {\path{doi:10.1145/800221.806707}}.

\bibitem{CKLS02}
Christian Cachin, Klaus Kursawe, Anna Lysyanskaya, and Reto Strobl.
\newblock Asynchronous verifiable secret sharing and proactive cryptosystems.
\newblock In {\em in Proc. 9th ACM Conference on Computer and Communications
  Security (CCS}, pages 88--97. ACM Press, 2002.

\bibitem{Cachin05}
Christian Cachin, Klaus Kursawe, and Victor Shoup.
\newblock Random oracles in constantinople: Practical asynchronous byzantine
  agreement using cryptography.
\newblock {\em Journal of Cryptology}, 18(3):219--246, 2005.
\newblock URL: \url{http://dx.doi.org/10.1007/s00145-005-0318-0}, \href
  {http://dx.doi.org/10.1007/s00145-005-0318-0}
  {\path{doi:10.1007/s00145-005-0318-0}}.

\bibitem{CGR07}
Tushar~D. Chandra, Robert Griesemer, and Joshua Redstone.
\newblock Paxos made live: An engineering perspective.
\newblock In {\em Proceedings of the Twenty-sixth Annual ACM Symposium on
  Principles of Distributed Computing}, PODC '07, pages 398--407, New York, NY,
  USA, 2007. ACM.
\newblock URL: \url{http://doi.acm.org/10.1145/1281100.1281103}, \href
  {http://dx.doi.org/10.1145/1281100.1281103}
  {\path{doi:10.1145/1281100.1281103}}.

\bibitem{CHT96}
Tushar~Deepak Chandra, Vassos Hadzilacos, and Sam Toueg.
\newblock The weakest failure detector for solving consensus.
\newblock {\em J. ACM}, 43(4):685--722, July 1996.
\newblock URL: \url{http://doi.acm.org/10.1145/234533.234549}, \href
  {http://dx.doi.org/10.1145/234533.234549} {\path{doi:10.1145/234533.234549}}.

\bibitem{CT96}
Tushar~Deepak Chandra and Sam Toueg.
\newblock Unreliable failure detectors for reliable distributed systems.
\newblock {\em J. ACM}, 43(2):225--267, March 1996.
\newblock URL: \url{http://doi.acm.org/10.1145/226643.226647}, \href
  {http://dx.doi.org/10.1145/226643.226647} {\path{doi:10.1145/226643.226647}}.

\bibitem{CS09}
Bernadette Charron-Bost and Andr{\'e} Schiper.
\newblock The heard-of model: computing in distributed systems with benign
  faults.
\newblock {\em Distributed Computing}, 22(1):49--71, 2009.
\newblock URL: \url{http://dx.doi.org/10.1007/s00446-009-0084-6}, \href
  {http://dx.doi.org/10.1007/s00446-009-0084-6}
  {\path{doi:10.1007/s00446-009-0084-6}}.

\bibitem{CIL94}
Benny Chor, Amos Israeli, and Ming Li.
\newblock Wait-free consensus using asynchronous hardware.
\newblock {\em SIAM J. Comput.}, 23(4):701--712, August 1994.
\newblock URL: \url{http://dx.doi.org/10.1137/S0097539790192635}, \href
  {http://dx.doi.org/10.1137/S0097539790192635}
  {\path{doi:10.1137/S0097539790192635}}.

\bibitem{CWAD09}
Allen Clement, Edmund Wong, Lorenzo Alvisi, Mike Dahlin, and Mirco Marchetti.
\newblock Making {B}yzantine fault tolerant systems tolerate {B}yzantine
  faults.
\newblock In {\em Proceedings of the 6th USENIX Symposium on Networked Systems
  Design and Implementation}, NSDI'09, pages 153--168, Berkeley, CA, USA, 2009.
  USENIX Association.
\newblock URL: \url{http://dl.acm.org/citation.cfm?id=1558977.1558988}.

\bibitem{Feldman87}
Paul Feldman.
\newblock A practical scheme for non-interactive verifiable secret sharing.
\newblock In {\em The 28th Annual Symposium on Foundations of Computer
  Science}, pages 427--438, Oct 1987.
\newblock \href {http://dx.doi.org/10.1109/SFCS.1987.4}
  {\path{doi:10.1109/SFCS.1987.4}}.

\bibitem{FLP85}
Michael~J. Fischer, Nancy~A. Lynch, and Michael~S. Paterson.
\newblock Impossibility of distributed consensus with one faulty process.
\newblock {\em J. ACM}, 32(2):374--382, April 1985.
\newblock URL: \url{http://doi.acm.org/10.1145/3149.214121}, \href
  {http://dx.doi.org/10.1145/3149.214121} {\path{doi:10.1145/3149.214121}}.

\bibitem{GW00}
Xianjun Geng and Andrew~B. Whinston.
\newblock Defeating distributed denial of service attacks.
\newblock {\em IT Professional}, 2(4):36--42, Jul 2000.
\newblock \href {http://dx.doi.org/10.1109/6294.869381}
  {\path{doi:10.1109/6294.869381}}.

\bibitem{HJKY95}
Amir Herzberg, Stanislaw Jarecki, Hugo Krawczyk, and Moti Yung.
\newblock Proactive secret sharing or: How to cope with perpetual leakage.
\newblock In {\em Proceedings of the 15th Annual International Cryptology
  Conference on Advances in Cryptology}, CRYPTO '95, pages 339--352, London,
  UK, UK, 1995. Springer-Verlag.
\newblock URL: \url{http://dl.acm.org/citation.cfm?id=646760.706016}.

\bibitem{JAD12}
Jafar~Haadi Jafarian, Ehab Al-Shaer, and Qi~Duan.
\newblock Openflow random host mutation: Transparent moving target defense
  using software defined networking.
\newblock In {\em Proceedings of the First Workshop on Hot Topics in Software
  Defined Networks}, HotSDN '12, pages 127--132, New York, NY, USA, 2012. ACM.
\newblock URL: \url{http://doi.acm.org/10.1145/2342441.2342467}, \href
  {http://dx.doi.org/10.1145/2342441.2342467}
  {\path{doi:10.1145/2342441.2342467}}.

\bibitem{KSMMZ03}
Sherif.~M. Khattab, Chatree Sangpachatanaruk, Rami Melhem, Daniel Mosse, and
  Taieb Znati.
\newblock Proactive server roaming for mitigating denial-of-service attacks.
\newblock In {\em International Conference on Information Technology: Research
  and Education (ITRE 2003)}, pages 286--290, Aug 2003.
\newblock \href {http://dx.doi.org/10.1109/ITRE.2003.1270623}
  {\path{doi:10.1109/ITRE.2003.1270623}}.

\bibitem{Lamport98}
Leslie Lamport.
\newblock The part-time parliament.
\newblock {\em ACM Trans. Comput. Syst.}, 16(2):133--169, May 1998.
\newblock URL: \url{http://doi.acm.org/10.1145/279227.279229}, \href
  {http://dx.doi.org/10.1145/279227.279229} {\path{doi:10.1145/279227.279229}}.

\bibitem{LMZ09}
Leslie Lamport, Dahlia Malkhi, and Lidong Zhou.
\newblock Vertical paxos and primary-backup replication.
\newblock In {\em Proceedings of the 28th ACM Symposium on Principles of
  Distributed Computing}, PODC '09, pages 312--313, New York, NY, USA, 2009.
  ACM.
\newblock URL: \url{http://doi.acm.org/10.1145/1582716.1582783}, \href
  {http://dx.doi.org/10.1145/1582716.1582783}
  {\path{doi:10.1145/1582716.1582783}}.

\bibitem{LNZBGS16}
Loi Luu, Viswesh Narayanan, Chaodong Zheng, Kunal Baweja, Seth Gilbert, and
  Prateek Saxena.
\newblock A secure sharding protocol for open blockchains.
\newblock In {\em Proceedings of the 2016 ACM SIGSAC Conference on Computer and
  Communications Security}, CCS '16, pages 17--30, New York, NY, USA, 2016.
  ACM.
\newblock URL: \url{http://doi.acm.org/10.1145/2976749.2978389}, \href
  {http://dx.doi.org/10.1145/2976749.2978389}
  {\path{doi:10.1145/2976749.2978389}}.

\bibitem{MJM08}
Yanhua Mao, Flavio~P. Junqueira, and Keith Marzullo.
\newblock Mencius: Building efficient replicated state machines for {WAN}s.
\newblock In {\em Proceedings of the 8th USENIX Conference on Operating Systems
  Design and Implementation}, OSDI'08, pages 369--384, Berkeley, CA, USA, 2008.
  USENIX Association.
\newblock URL: \url{http://dl.acm.org/citation.cfm?id=1855741.1855767}.

\bibitem{NvR15}
Stavros Nikolaou and Robbert Van~Renesse.
\newblock Turtle consensus: Moving target defense for consensus.
\newblock In {\em Proceedings of the 16th Annual Middleware Conference},
  Middleware '15, pages 185--196, New York, NY, USA, 2015. ACM.
\newblock URL: \url{http://doi.acm.org/10.1145/2814576.2814811}, \href
  {http://dx.doi.org/10.1145/2814576.2814811}
  {\path{doi:10.1145/2814576.2814811}}.

\bibitem{OY91}
Rafail Ostrovsky and Moti Yung.
\newblock How to withstand mobile virus attacks (extended abstract).
\newblock In {\em Proceedings of the Tenth Annual ACM Symposium on Principles
  of Distributed Computing}, PODC '91, pages 51--59, New York, NY, USA, 1991.
  ACM.
\newblock URL: \url{http://doi.acm.org/10.1145/112600.112605}, \href
  {http://dx.doi.org/10.1145/112600.112605} {\path{doi:10.1145/112600.112605}}.

\bibitem{Shamir79}
Adi Shamir.
\newblock How to share a secret.
\newblock {\em Commun. ACM}, 22(11):612--613, November 1979.
\newblock URL: \url{http://doi.acm.org/10.1145/359168.359176}, \href
  {http://dx.doi.org/10.1145/359168.359176} {\path{doi:10.1145/359168.359176}}.

\bibitem{Shoup00}
Victor Shoup.
\newblock Practical threshold signatures.
\newblock In {\em Proceedings of the 19th International Conference on Theory
  and Application of Cryptographic Techniques}, EUROCRYPT'00, pages 207--220,
  Berlin, Heidelberg, 2000. Springer-Verlag.
\newblock URL: \url{http://dl.acm.org/citation.cfm?id=1756169.1756190}.

\bibitem{emulab02}
Brian White, Jay Lepreau, Leigh Stoller, Robert Ricci, Shashi Guruprasad, Mac
  Newbold, Mike Hibler, Chad Barb, and Abhijeet Joglekar.
\newblock An integrated experimental environment for distributed systems and
  networks.
\newblock In {\em Proceedings of the 5th Symposium on Operating Systems Design
  and Implementation (OSDI'02)}, pages 255--270, Boston, MA, December 2002.
  Usenix.

\bibitem{ZSV05}
Lidong Zhou, Fred~B. Schneider, and Robbert Van~Renesse.
\newblock Apss: Proactive secret sharing in asynchronous systems.
\newblock {\em ACM Trans. Inf. Syst. Secur.}, 8(3):259--286, August 2005.
\newblock URL: \url{http://doi.acm.org/10.1145/1085126.1085127}, \href
  {http://dx.doi.org/10.1145/1085126.1085127}
  {\path{doi:10.1145/1085126.1085127}}.

\end{thebibliography}
